\documentclass{article}

\usepackage[utf8]{inputenc}
\usepackage[a4paper,top=3cm,bottom=3cm,left=2cm,right=2cm,marginparwidth=1.75cm]{geometry}
\usepackage{amsmath,amsfonts,amssymb,amsthm}
\usepackage{enumerate}
\usepackage{xcolor}
\usepackage{float}
\usepackage{graphicx}
\usepackage{multirow}
\usepackage{hyperref} 
\usepackage{array} 
\usepackage{placeins}

\newcommand{\Esp}{\mathbb{E}}
\newcommand{\Mcal}{\mathcal{M}}

\renewcommand{\Pr}{\mathbb{P}}
\newtheorem{rem}{Remark}

\newtheorem{algorithm}{Algorithm}
\newtheorem{Prop}{Proposition}

\newcommand{\SR}[2]{\textcolor{gray}{#1}\textcolor{blue}{#2}}

\title{A robust model-based clustering based on the geometric median and the Median Covariation Matrix}
\author{Antoine Godichon-Baggioni$^{1*}$, Stéphane Robin$^{1}$}
\date{}

\begin{document}
\maketitle

$^1$ : Sorbonne Université, CNRS, Laboratoire de Probabilités, Statistique et Modélisation, F-75005 Paris, France

$^*$ : \url{antoine.godichon\_baggioni@upmc.fr}

\begin{abstract}
Grouping observations into homogeneous groups is a recurrent task in statistical data analysis. We consider Gaussian Mixture Models, which are the most famous parametric model-based clustering method. We propose a new robust approach for model-based clustering, which consists in a modification of the EM algorithm (more specifically, the M-step) by replacing the estimates of the mean and the variance by robust versions based on the median and  the median covariation matrix.  All the proposed methods are available in the R package \texttt{RGMM} accessible on CRAN.
\end{abstract}


\section{Introduction} \label{sec:intro}
\paragraph{Problem.}
Grouping observations into homogeneous groups (or "clusters") is one of the most typical tasks in statistical data analysis. Among the many methods that have been proposed over the years, model-based clustering is one of the most popular (\cite{MaP00}). Model-based clustering relies on the assumption that the observed data come from a mixture model, meaning that the observations can be divided into a finite (but often unknown) number of clusters, and that each cluster is characterized by a specific distribution, often called the {\em emission} distribution.

One reason for the popularity of model-based clustering is that the emission distributions of the clusters are usually chosen with a parametric class (e.g. a multivariate Gaussian), which makes the interpretation of the results particularly easy. Another reason for this popularity is that the maximum likelihood estimates of the parameters can be obtained via the well-known EM algorithm (\cite{DLR77}), accompanied by statistical guarantees. 

Nevertheless, one of the weaknesses of model-based clustering methods is their sensitivity to misspecification of emission distributions or to the presence of (possibly numerous) outliers. In both cases, this results in a high proportion of misclassified observations  or a poor estimate of the number of clusters \cite{GGM10}.


\paragraph{Robust approaches.} 
A series of robust approaches have been proposed to overcome these limitations. These approaches can be classified into three main categories. 
A first track sticks to the parametric framework, but uses emission distributions with heavier tails, {such as multivariate student for Gaussian mixtures} (see, e.g., \cite{PeM00,WaL15,SPI15,RoS19}). Alternatively, a component associated with (possibly improper) parametric distribution can be added, in order to capture outliers (\cite{BaR93,CoH16,CoH17,FaP20}). The outlier distribution may typically be uniform emission over a large domain.
A second approach is to prune the observations, so that the outliers do not weigh too heavily on the estimates \cite{GGM08}. A final approach is to use a dedicated weighted contrast (instead of negative log-likelihood: \cite{GYZ19,GMY21}). The latter approach has some similarities with {the method we propose}.

\paragraph{Our contribution.}
This paper focuses on the robustness of model-based clustering methods to the presence of outliers, meaning that we make no assumptions about how outliers deviate from prescribed emission distributions. To this end, we adopt a fully parametric model-based clustering framework, but modify the EM algorithm (more specifically, the M-step) to ensure robustness. Our   method is valid for any symmetric emission distribution and resorts to the estimation of the median vector and the median covariation matrix  {in place} of the mean vector and the covariance matrix. The estimation of these quantities benefits from a series of recent contributions \cite{VZ00,HC,CG2015}. It was especially proven (see \cite{KrausPanaretos2012}) that for symmetric distributions, the MCM and the usual covariance have the same eigenvectors. Nevertheless, although the recursive estimation of the MCM has been studied in \cite{CG2015}, no method for building the covariance from the MCM has been proposed. In this paper, we first propose methods to get  robust estimates of the covariance  before applying it to robust model-based clustering. 

\paragraph{Outline.}
 The following section gives a comparative introduction to recent algorithms for estimating median vectors and median covariation matrices. In Section \ref{sec:RMM}, we show how these estimates can be used for robust inference of mixture models. The resulting algorithm is given in Section \ref{sec:algo}. A comprehensive simulation study is presented in Section \ref{sec:simul}: different estimators of the median covariation matrix are first compared in Section \ref{sec:simVariance}, then the classification accuracy of the proposed EM-type algorithm is evaluated in Section \ref{sec:simMixture}. All the proposed methods are available in the R package \texttt{RGMM} accessible on CRAN\footnote{\url{cran.r-project.org/package=RGMM}}.

\section{Robust estimation of the emission parameters} \label{sec:variance}
\subsection{Estimating the geometric median}
In {this section, we} consider a random variable $X$  {with} values in $\mathbb{R}^{d}$. The geometric median of $X$ is defined (see \cite{Hal48,Kem87}) by 
\[
m^{*} \in \arg\min_{m \in \mathbb{R}^{d}} \mathbb{E} \left[ \left\| X - m \right\| - \| X \| \right]
\]
where $\|\cdot\|$ stands for the $\ell_2$ norm.
Remark that the term $\| X \| $ just enables not to have to make any assumption on the existence of the first order moment of the  random vector $X$. If the random variable $X$ is not concentrated on a straight line nor single points, the geometric median is uniquely defined \cite{Kem87}. An iterative way to estimate the median, giving i.i.d. copies $X_{1} , \ldots ,X_{n}$ of $X$, is to consider the median as a fix point, leading to the following Weiszfeld algorithm \cite{Weiszfeld1937,VZ00}:
\begin{equation}
\label{def:Weisz} m_{t+1} 
{
= \frac{\sum_{k=1}^{n} {X_{k}}/{\left\| X_{k} - m_{t} \right\|}}{\sum_{k=1}^{n} {1}/{\left\| X_{k} - m_{t} \right\|}}
}
.
\end{equation}
A recursive and faster way (in term of computational cost) to estimate the median is to consider the Averaged Stochastic Gradient (ASG) algorithm (\cite{HC,CCG2015,godichon2015}) defined recursively for all $k \leq n-1$ by
\begin{align*}
m_{k+1} & = m_{k} + \gamma_{k+1} \frac{X_{k+1} - m_{k}}{\left\| X_{k+1} - m_{k} \right\|} \\
\overline{m}_{k+1} & = \overline{m}_{k} + \frac{1}{k+1} \left( m_{k+1} - \overline{m}_{k} \right)
\end{align*}
where $m_{0} = \overline{m}_{0}$ is arbitrarily chosen, $\gamma_{k} = c_{\gamma}k^{-\gamma}$ with $c_{\gamma} > 0 $ and $\gamma \in (1/2,1)$. Remark that under weak assumptions, these estimates (ASG and Weiszfeld) are asymptotically efficient (\cite{HC,VZ00}). Nevertheless, in case of small samples lying in moderate dimension spaces, one should prefer Weiszfeld algorithm and vice versa.

\begin{rem}
Remark that for mixture model, we will consider a weighted version of the median, i.e considering a positive random variable $w$, we will consider
\[
m^{*} = \arg\min_{m} \mathbb{E}\left[ w \left\| X - m \right\| - w \| X \| \right]
\]
leading, considering $\left( X_{1} , w_{1} \right) , \ldots , \left( X_{n},w_{n} \right)$, to the following transformation of the algorithms
\[
m_{t+1} 
= 
{
\frac{\sum_{k=1}^{n} {w_{k}X_{k}} / {\left\| X_{k} - m_{t} \right\|}}{\sum_{k=1}^{n} {w_{k}}/ {\left\| X_{k} - m_{t} \right\|}}
}
\qquad \text{and} \qquad m_{k+1}  = m_{k} + \gamma_{k+1}w_{k+1} \frac{X_{k+1} - m_{k}}{\left\| X_{k+1} - m_{k} \right\|} .
\]
\end{rem}

\subsection{Estimating the Median Covariation Matrix}

The Median Covariation Matrix (MCM for short) is defined (\cite{KrausPanaretos2012}, \cite{CG2015}) by 
\[
V^{*} \in \arg\min_{V \in \mathcal{M}_{d}\left( \mathbb{R} \right)} \mathbb{E}\left[ \left\| \left( X - m^{*} \right)\left( X - m^{*} \right)^{T} - V \right\|_{F} - \left\| \left( X - m^{*} \right)\left( X - m^{*} \right)^{T} \right\|_{F} \right]
\]
where $m^{*}$ is the geometric median of $X$, $\mathcal{M}_{d}\left( \mathbb{R} \right)$ denotes the vectorial space of squared real matrices of size $d \times d$ and $\| . \|_{F}$ is the associated Frobenius norm. In other words, the MCM can be seen as the geometric median of the random matrix $\left( X - m^{*} \right) \left( X - m^{*} \right)^{T}$. Then, given the estimate $m_{T}$ of $m^{*}$ obtained with \eqref{def:Weisz} after $T$ iterations, one can consider the Weiszfeld algorithm \cite{CG2015}
$$
V_{t+1} = 
{
\frac{\sum_{k=1}^{n} {\left\| \left( X_{k} - m_{T} \right) \left( X_{k} - m_{T} \right)^{T} - V_{t} \right\|^{-1}_{F}} {\left( X_{k} - m_{T} \right) \left( X_{k} - m_{T} \right)^{T}} }{\sum_{k=1}^{n} {\left\| \left( X_{k} - m_{T} \right) \left( X_{k} - m_{T} \right)^{T} - V_{t} \right\|^{-1}_{F} } }
}.
$$
In the same way, one can both estimate the median and the MCM recursively considering the ASG algorithm
\begin{align*}
V_{k+1} & = V_{k} + \gamma_{k+1} \frac{\left( X_{k+1} - \overline{m}_{k} \right)\left( X_{k+1} - \overline{m}_{k}\right)^{T} - V_{n}}{\left\| \left( X_{k+1} - \overline{m}_{k} \right)\left( X_{k+1} - \overline{m}_{k}\right)^{T} - V_{n} \right\|_{F}} \\
 \overline{V}_{k+1} & = \overline{V}_{k} + \frac{1}{k+1}\left( V_{k+1} - \overline{V}_{k} \right) , 
\end{align*}
with $\overline{V}_{0} = V_{0}$ symmetric and positive. {First observe} that the estimates are not necessarily positive, but one can project {them}  onto the set of definite positive matrices or consider the modification of the stepsequence proposed in \cite{CG2015}. 
Remark  {also} that, here again, one can consider the weighted version of the MCM and modify the algorithm {accordingly}.

\subsection{Robust estimation of the variance}
Let us  {now suppose} that $X$ admits a second order moment and   denote by $\mu$ and $\Sigma $ its mean and variance (supposed to be positive). All this work relies on the fact that, if the distribution of $X$ is symmetric, $V^{*}$ and $\Sigma$ have the same eigenvectors (\cite{KrausPanaretos2012}). Furthermore, denoting $U=\left( U_{1} , \ldots , U_{d} \right)^{T}  := \Sigma^{-1/2} \left( X - \mu \right)$ and $\delta$ (resp. $\lambda$) the vector of eigenvalues (by decreasing order) of $V^{*}$ (resp. $\Sigma$), one has (\cite{KrausPanaretos2012}): 
\begin{equation}
\label{link:delta:lambda} 
 {\delta_{k}} = \lambda_{k} \mathbb{E}\left[ U_{k}^{2} h\left( \delta , \lambda , U \right)   \right]\left( \mathbb{E}\left[ h \left( \delta , \lambda ,  U \right)  \right] \right)^{-1}
\end{equation}  
where $h(\delta , \lambda , U ) := \left( \sum_{i=1}^{d} \left(  {\delta_{i}} - \lambda_{i}U_{i}^{2} \right)^{2} + \sum_{i \neq j} \lambda_{i}\lambda_{j} U_{i}^{2}U_{j}^{2}\right)^{-1/2}$. In what follows, we will denote by $\Psi_{U}$ the function such that
\begin{equation}
\label{def:Psi} \Psi_{U} \left( V^{*} \right) = \Sigma.
\end{equation}
 Let us suppose from now that the law of $U$ is known and that we know how to simulate i.i.d random variables following this law. For instance, for the Gaussian case, it is  {clear} that $U \sim \mathcal{N}\left( 0 ,I_{d} \right)$. In a same way, for the multivariate Student with $p$ degrees of freedom (with $p \geq 3$), one has $U = \sqrt{p-2} {{N_{d}} / {\sqrt{K_{p}}}}$ where $N_{d} \sim \mathcal{N}\left(0 , I_{d} \right)$ and $ K_{p} \sim \chi_{p}^{2}$ are independent. Finally, for the multivariate case, $U$ follows a standard multivariate law. Let us now consider i.i.d. copies of $X$ and an estimate of the MCM $V_{n}$. Let us denote by $\delta_{n} = \left( \delta_{1,n} , \ldots , \delta_{d,n} \right)$ and $\left( v_{1,n} , \ldots , v_{p,n} \right)$ the eigenvalues (by {decreasing} order) and the associated eigenvectors,  {respectively}. 
  {Robust estimates of the eigenvalues of the variance can be obtained via a Monte-Carlo approach, based on $N$ i.i.d. copies $U_{1},\ldots ,U_{N}$ of $U$.} 
 A first solution to estimate $\lambda$ is so to consider the following fix point algorithm:
\begin{algorithm}[Fix point algorithm]
For all $t \in \mathbb{N}$, and $k=1, \ldots ,d$,
\[
\lambda_{n,N,t+1}[k] = \delta_{n}[k] \frac{\sum_{i=1}^{N} h \left( \delta_{n}, \lambda_{n,N,t},U_{i} \right) }{\sum_{i=1}^{N} \left( U_{i}[k] \right)^{2} h \left( \delta_{n}, \lambda_{n,N,t},U_{i} \right)}
\]
where for all $x= \left( x_{1} , \ldots , x_{d} \right)^{T} \in \mathbb{R}^{d}$, $x[k] = x_{k}$. 
\end{algorithm}

\noindent Remark that this method does not require to calibrate any hyperparameter, as opposed to the possibly more efficient following gradient algorithm.

\begin{algorithm}[Gradient algorithm]
For all $t \in \mathbb{N}$,
\[
\lambda_{n,N,t+1} = \lambda_{n,N,t} - \eta_{t}\sum_{k=1}^{n} \lambda_{n,N,t}\left(  U_{i}^{2}  h\left( \delta_{n}, \lambda_{n,N,t},U_{i} \right) -  \delta_n h   \left( \delta_{n}, \lambda_{n,N,t},U_{i} \right) \right)
\]
where $\eta_{t}$ is non-decreasing positive step sequence.
\end{algorithm}

\noindent One may refer to the classical literature to calibrate the step sequence. 
Finally, one may resort to the recursive Robbins-Monro  algorithm \cite{robbins1951} or, more specifically, to its weighted averaged version \cite{mokkadem2011generalization}.
\begin{algorithm}[Robbins-Monro]
For all $k \leq N-1$, one has
\begin{align*}
\lambda_{n,N,k+1} & = \lambda_{n,N,k} - \gamma_{k+1}  \left( \lambda_{n,N,k} U_{k+1}^{2}  h\left( \delta_{n}, \lambda_{n,N,k},U_{k+1} \right) -  \delta_{n} h   \left( \delta_{n}, \lambda_{n,N,k},U_{k+1} \right) \right), \\
\overline{\lambda}_{n,N,k+1} & = \overline{\lambda}_{n,N,k} + \frac{\log (k+1)^{w}}{\sum_{l=0}^{k}\log(l+1)^{w}}\left( \lambda_{n,N,k+1} - \overline{\lambda}_{n,N,k} \right),
\end{align*}
with $\overline{\lambda}_{n,N,0}=\lambda_{n,N,0}$, $\gamma_{k} = c_{\gamma}k^{-\gamma}$ with $c_{\gamma}> 0$ and $\gamma \in (1/2,1)$, $\omega \geq 0$.
\end{algorithm}

\noindent 
{
The term 'weighted averaging' comes from the update formula $\overline{\lambda}_{n,N,k} = \sum_{l=0}^{k}\log(l+1)^{\omega}  {\lambda}_{n,N,l} \left/ {\sum_{l=0}^{k}\log(l+1)^{\omega}} \right.$. 
}
Note that the case  $w = 0$ corresponds to the usual averaged algorithm \cite{ruppert1988efficient,PolyakJud92}. 

Remark that, under suitable assumptions, these three methods have the same asymptotic (on $N$) behavior. Nevertheless, the two first ones necessitate $O\left(Nd^{2}T \right)$ operations, where $T$ is the number of iterations, while the Weighted Averaged Robbins Monro algorithm  {only requires} $O \left( Nd^{2} \right)$ operations.  {Hence, using $TN$ copies instead of $N$, one expect a better precision with the last method, for the same computational time.}
A comparative study of these algorithms  {is presented} in Section \ref{sec:simVariance}.

\section{Robust Mixture Model} \label{sec:RMM}
\subsection{Mixture model}
We now consider a random variable $X$ following a mixture with $K$ classes, i.e
\begin{equation}
\label{def:RMM} X \sim \sum_{k=1}^K \pi_k^{*} Y_{k},
\end{equation}
that is $Z \sim \Mcal(1, \pi^{*})$ and $(X \mid Z=k) \sim Y_{k}$, where the vector $\pi^{*} = \left( \pi_{1}^{*} , \ldots , \pi_{K}^{*} \right)$ belongs to the simplex $ {\mathcal{S}}^{K}:= \left\{ \pi, \pi_{k} > 0, \sum_{k=1}^{K} \pi_{k} = 1 \right\}$.

Furthermore, we suppose from  now that $Y_{k}$ satisfies the following conditions:
\begin{enumerate}[($a$)]
\item the distribution of $Y_{k}$ is symmetric;
\item $Y_{k}$ admits a second order moment, and we denote by $\mu_{k}^{*}$ and $\Sigma_{k}^{*}$ its mean and variance;
\item the variance of $Y_{k}$ is positive;
\item the random variable $Y_{k}$ is absolutely continuous with density $\phi_{\mu_{k}^{*},\Sigma_{k}^{*}}(.)$ determined by $\mu_{k}^{*}, \Sigma_{k}^{*}$ and known parameters. 
\end{enumerate}

 Observe that these conditions are satisfied by Gaussian mixtures, multivariate Student mixtures (when all the classes have the same known degree of freedom) or multivariate Laplace mixtures (to name a few).  {Conditions ($a$), ($b$) and ($c$)} enable to  {estimate} the mean and the variance  {in a robust manner} with the methods proposed in previous section, while  {Condition ($d$)} just ensures that the density only depends on known parameters or on parameters that can be estimated robustly. Of course, one can adapt this work for more specific cases such as Student mixtures with unknown degrees of freedom. In what follows, we will denote $\mu^{*} = \left( \mu_{1}^{*} , \ldots , \mu_{K}^{*} \right)$, $\Sigma^{*} = \left( \Sigma_{1}^{*} , \ldots , \Sigma_{K}^{*} \right)$ and $\theta^{*} = (\pi^{*} , \mu^{*} , \Sigma^{*})$. 

The popular EM algorithm (\cite{DLR77}) aims at providing the maximum likelihood estimates by minimizing the empirical risk
\[
R_n \left(\pi, \mu , \Sigma \right) - \frac1n \sum_{i=1}^n \sum_{k=1}^{K} \tau_{k}(X_i)  \left(\log \left( \pi_{k} \right) + \log \left( \phi_{\mu_{k}, \Sigma_{k} }\left( X_i \right) \right) \right),
\]
the theoretical counterpart of which is
\[
R \left(\pi, \mu , \Sigma \right) = - \mathbb{E}_{\theta^*}\left[ \sum_{k=1}^{K} \tau_{k}(X)  \left(\log \left( \pi_{k} \right) + \log \left( \phi_{\mu_{k}, \Sigma_{k} }\left( X \right) \right) \right) \right] ,
\]
where 
$
\tau_k(X) 
= \Pr_{\theta^{*} }\left[ Z = k \mid X \right]
{
= {\pi_k^{*} \phi_{\mu_{k}^{*} , \Sigma_{k}^{*}}(X)} \left/ {\sum_{\ell = 1}^K \pi_\ell^{*} \phi_{\mu_\ell^{*},\Sigma_{\ell}^{*}}(X)} \right.
}
$
.

 {Importantly, we now that}
\[\pi^{*} \in \arg\min_{\pi \in  \mathcal{S}^{K}} - \Esp_{\theta^{*}}\left[\sum_{k=1}^K {\tau}_k(X) \log \pi_k\right]
\]
while 
$$
\mu^{*} = \arg\min_{\mu} \mathbb{E}_{\theta^{*}} \left[ \sum_{k=1}^{K} \tau_{k}(X) \left\| X - \mu_k \right\|^{2} \right],  
\qquad 
\Sigma^{*} = \arg\min_{\Sigma} \mathbb{E}_{\theta^{*}}\left[ \sum_{k=1}^{K} \tau_{k}(X) \left\| \left( X - \mu^{*} \right) \left( X - \mu^{*} \right)^{T} - \Sigma_k \right\|_{F}^{2} \right] , 
$$
where $\| .\|_{F}$ is the Frobenius norm for matrices. 


\subsection{Loss}
Consider a mixture model as defined in \eqref{def:RMM} with parameter $\theta^{*} = (\pi^{*} , \mu^{*} , \Sigma^{*})$ and let us denote by $m^{*} = \left( m_{1}^{*} , \ldots , m_{K}^{*} \right)$ and $V^{*} = \left( V_{1}^{*}, \ldots , V_{K}^{*} \right)$ the medians and MCM of the classes. Intuitively, the idea is to replace, in the usual EM algorithm, the estimates of the mean $\mu_{k}$ and the variance $\Sigma_{k}$ of each  class by their robust version. More precisely, the aim is to replace them by the median $m_{k}^{*}$ and the transformation of the MCM $\Psi_{U}\left( V_{k}^{*} \right)$ of each class.
 {Still}, as we cannot know the class of the data,  {we need to} give an alternative definition of the median and MCM of the classes.   To do so, let us introduce the two following functions:
\begin{align*}
   & G_2(m)  = \Esp_{\theta^{*} }\left[\sum_{k=1}^K \tau_k(X) \left\|X - m_k\right\|\right] 
   & G_3(m, V) 
    = \Esp_{\theta^{*} }\left[\sum_{k=1}^K  \tau_k(X) \left\|(X - m_k)(X - m_k)^\intercal - V_k\right\|_{F}\right]. 
\end{align*}

The following proposition ensures that the minimizers of these functions correspond to $m^{*}$ and $V^{*}$, which will be  {central in the construction of} robust estimates of $\theta^{*}$.
\begin{Prop}\label{prop1}
  Consider a mixture model as defined in \eqref{def:RMM} and parametrized with $\theta^* = (\pi^*, \mu^*, \Sigma^*)$. Then
  $$    m^* = \arg\min_m \Esp_{\theta^*}  \left[ G_2( m)\right], \quad \quad \text{and} \quad \quad
  V^* = \arg\min_V \Esp_{\theta^*} \left[G_3( m^*, V)\right].
  $$
  Furthermore, $m^{*} = \mu^{*}$, $\Psi_{U} \left( V^{*} \right) := \left( \Psi_{U} \left( V_{1}^{*} \right) , \ldots , \Psi_{U} \left( V_{K}^{*} \right) \right) = \Sigma^{*}$,
  $$
  \tau_{k}(X) = {\pi_k^{*} \phi_{m_{k}^{*} , \Psi_{U}\left( V_{k}^{*} \right)}(X)}\left/ {\sum_{\ell = 1}^K \pi_\ell^{*} \phi_{m_\ell^{*},\Psi_{U}\left( V_{\ell}^{*} \right)}(X)} \right.,
  $$
  and 
  \[
  R_{\pi^{*}} \left(   m^{*} , \Psi\left( {V}^{*} \right) \right) = \min_{\mu ,\Sigma} R_{\pi^{*}}\left( \mu , \Sigma \right) =  R_{\pi^{*}}\left( \mu^{*} , \Sigma^{*} \right). 
  \]
\end{Prop}
In other words, we propose here a new parametrization of the problem where the new parameters correspond to robust indicators.
The proof is given in Section \ref{sec:proofs}.

\subsection{Fix-point property} 
The following proposition enables to see $\left( \pi^{*},m^{*},V^{*} \right)$ as a fixpoint of a function $g^{*}$.
\begin{Prop}\label{prop:fixpoint}
    Consider a mixture model as defined in \eqref{def:RMM} and parametrized with $\theta^* = (\pi^*, \mu^*, \Sigma^*)$. Then, $\left( \pi^{*},m^{*},V^{*} \right)$ (with $\pi^{*},m^{*},V^{*}$ defined in Proposition \ref{prop1}) satisfy
\[
\left( \pi^{*},m^{*},V^{*} \right) = g^{*} \left( \pi^{*},m^{*},V^{*} \right)    
\]
where $g^{*}(\pi , m , V) = \left( g_{1}^{*}(\pi),g_{2,1}^{*}\left(m_{1}\right)),\ldots ,g_{2,K}^{*}\left(m_{K}\right) ,g_{3,1}^{*}\left(V_{1}\right),g_{3,K}^{*}\left(V_{K}\right) \right)$ with $g_{1}(\pi) = \left( g_{1,1} \left( \pi \right) , \ldots , g_{1,K} \left( \pi \right) \right) $ and
{
\begin{align*}
  g_{1,k}(\pi) & := \mathbb{E}\left[ \frac{ \pi_{k} \phi \left( X , m_{k}^{*} , \Psi_{U} \left( V_{k}^{*} \right) \right) }{\sum_{i=1}^{K}\pi_{i}  \phi \left( X , m_{i}^{*} , \Psi_{U} \left( V_{i}^{*} \right)  \right)} \right] \qquad
  g_{2,k}\left( m_{k} \right) := \frac{\mathbb{E}\left[  \frac{\tau_{k}(X)X}{\left\| X - m_{k}  \right\|} \right]   }{ \mathbb{E}\left[  \frac{\tau_{k}(X)}{\left\| X - m_{k}  \right\|} \right]}   \\
  g_{3,k} \left( V_{k}  \right) & := {\mathbb{E}\left[  \frac{\tau_{k}(X) \left( X - m_{k}^{*} \right) \left( X - m_{k}^{*} \right)^{T} }{\left\| \left( X - m_{k}^{*} \right) \left( X - m_{k}^{*} \right)^{T} - V_{k}  \right\|_{F}} \right]   } \left( { \mathbb{E}\left[  \frac{\tau_{k}(X)  }{\left\| \left( X - m_{k}^{*} \right) \left( X - m_{k}^{*} \right)^{T} - V_{k}  \right\|_{F}} \right]} \right)^{-1}  
\end{align*}
}
and  {$\tau_{k}(X) = {\pi_k^{*} \phi_{m_{k}^{*} , \Psi_{U}\left( V_{k}^{*} \right)}(X)} \left/ {\sum_{\ell = 1}^K \pi_\ell^{*} \phi_{m_\ell^{*},\Psi_{U}\left( V_{\ell}^{*} \right)}(X)} \right.$}.
\end{Prop}

\section{Algorithm} \label{sec:algo}
\subsection{The algorithm}

 {We} consider that $X$ follows the mixture model defined by \eqref{def:RMM} and consider $X_{1} , \ldots ,X_{n}$ i.i.d copies of $X$.  We now consider the "empirical fixpoint function", i.e we will consider, denoting $\tau = \left( \tau_{1} , \ldots , \tau_{k} \right)$, and $\tau_{k} = \left( \tau_{1,k} , \ldots , \tau_{n,k} \right)$,
{
\begin{align*}
\hat{g}_{2,k} \left( \tau_{k} ,  m_{k} \right)
& = \left({\sum_{i=1}^{n} \frac{\tau_{i,k} X_{i}}{\left\| X_{i} - m_{k} \right\|} } \right) \left/ \left({\sum_{i=1}^{n} \frac{\tau_{i,k}}{\left\| X_{i} - m_{k} \right\|}} \right) \right. \\
\hat{g}_{3,k} \left( \tau_{k} ,  m_{k} , V_{k} \right)
& = \left( {\sum_{i=1}^{n}\frac{ \tau_{i,,k} \left( X_{i} - m_{k} \right) \left( X_{i} - m_{k} \right)^{T}}{\left\| \left( X_{i} - m_{k} \right) \left( X_{i} - m_{k} \right)^{T} - V_{k} \right\|_{F}}} \right) \left/ \left({\sum_{i=1}^{n} \frac{\tau_{i,k}}{\left\| \left( X_{i} - m_{k} \right) \left( X_{i} - m_{k} \right)^{T} - V_{k} \right\|_{F}}} \right) \right..
\end{align*}
}
This leads to the following algorithm:

\begin{algorithm}[Fix Point algorithm for Robust Mixture Model]
  Starting from $\phi^0 = (\pi^0, m^0, V^0)$, repeat until convergence:
  \begin{enumerate}
  \item Compute for each $1 \leq i \leq n$ and $1 \leq k \leq K$
  $$
  \tau_k^{h+1}(X_i) = \frac{\pi_{k}^{h}\phi_{m_{k}^{h}, \hat{\Psi}_{u} \left( V_{k}^{h} \right)} \left( X_{i} \right)}{\sum_{\ell = 1}^{K} \pi_{\ell}^{h}\phi_{m_{\ell}^{h}, \hat{\Psi}_{u} \left( V_{\ell}^{h} \right)} \left( X_{i} \right)} ,
  $$
  where $\hat{\Psi}_{U}$ is the application which enables, given $V_{k}$, to "rebuild" $\Sigma_{k}$ with the help of one of the method proposed in Section \ref{sec:variance}; 
  \item Based on the fix point relations (see Proposition \eqref{prop:fixpoint}), update, for each $1 \leq k \leq K$,
  $$
  \pi^{h+1}_k = \frac1n \sum_{i=1}^n \tau_k^{h+1}(X_i), \qquad
  m^{h+1}_k = \text{FixPoint}\left( \widehat{g}_{2k}(\tau_{k}^{h},.) \right), \qquad
  V^{h+1}_k =\text{FixPoint}\left( \widehat{g}_{3k}(\tau_{k}^{h},m_{k}^{h},.)\right).
  $$
  where $\text{FixPoint}\left(f(.) \right)$ denotes the fix point of the functional $f$.
  \end{enumerate}
\end{algorithm}
Note that estimating the fix points leads to estimate the weighted median and MCM considering weights $\tau_{k}^{h}$. More intuitively, this algorithm consists in updating $\tau_{i,k}$ replacing the empirical mean and variance of each class by their robust estimates based on the median and the MCM of each class, before updating $\pi$ (as usually).

\subsection{Choosing the number of clusters}

To determine the number of clusters $K$, we resort to two standard penalized-likelihood criteria, namely BIC (\cite{Sch78}) and ICL (\cite{BCG00,MaP00}).
More specifically, denoting by $D_K$ the number of independent parameters involved in a mixture with $K$ clusters and by $\widehat{\mathcal{L}}_K(X)$ the log-likelihood of the dataset $X$ evaluated with the parameter estimates resulting from the proposed estimation procedure:
$$
\widehat{\mathcal{L}}_K(X) = \sum_{i=1}^n \log\left(\sum_{k=1}^K \widehat{\pi}_k \phi_{\widehat{\mu}_k, \widehat{\Sigma}_k}(X_i)\right), 
$$
we used
\begin{equation} \label{eq:modelSel}
  BIC(K) = \widehat{\mathcal{L}}_K(X) - \log(n) D_K/2, \qquad 
  ICL(K) = BIC(K) + \sum_{i=1}^n \sum_{k=1}^K \widehat{\tau}_{i, k} \log \widehat{\tau}_{i, k}.
\end{equation}
We remind that the additional penalty term in the ICL criterion corresponds to the entropy of the conditional distribution of the latent variables $\{Z_i\}_{1 \leq i \leq n}$, conditional on the observed ones $\{X_i\}_{1 \leq i \leq n}$.  {This additional penalty is supposed to favor clusterings with lower classification uncertainty.}

\subsection{Initialization of the algorithm}
 {
We considered two ways of initializing the algorithm:
\begin{enumerate}
\item[•] Use the robust hierarchical clustering proposed by \cite{gagolewski2016genie}, to get $\tau^{1}$, and run our algorithm from there ;
\item[•] Randomly choose $K$ centers from the data and take $\Sigma_{k} = I_{d}$ and $\pi_{k} = \frac{1}{K}$ for all $k$. 
\end{enumerate}
Remark that the later way can tried several times, so to keep initialization leading to the best final log-likelihood.
We may also use the two ways and keep the best result in term of log-likelihood. 
}

\section{Simulations}  \label{sec:simul}
We designed a series of simulation studies to assess the efficiency and the accuracy of the proposed methodology.
 {The proposed methods are all} available in the R package \texttt{RGMM} available on CRAN\footnote{\url{cran.r-project.org/package=RGMM}}.

\subsection{Simulation design} \label{sec:simDesign}

\paragraph{Simulation parameters.}
We considered random vectors $X$ with dimension $p=5$ and mixture models with $K = 3$ clusters with equal proportions. We defined the three mean vectors $\mu_1$, $\mu_2$ and $\mu_3$, each with their $p$ coordinates all equal to $0$, $3$ and $-3$, respectively (see Equation \eqref{eq:mu} in Appendix \ref{app:simDesign}). We defined the three covariance matrices $\Sigma_1$, $\Sigma_2$ and $\Sigma_3$ given in Equation \eqref{eq:Sigma} in Appendix \ref{app:simDesign}.  {To give different shape to the three distributions}, $\Sigma_1$ has a constant diagonal term (equal to $2$), $\Sigma_2$ has diagonal terms increasing from $1$ to $p$ and $\Sigma_3$ has diagonal terms decreasing from $1$ to $1/p$. The two considered mixture distributions were therefore 
\begin{align*}
  \text{Gaussian: } & K^{-1} \sum_{k=1}^K \mathcal{N}_p\left(\cdot; \mu_k^*, \Sigma_k^*\right), &
  \text{Student: } & K^{-1} \sum_{k=1}^K \mathcal{T}_p\left(\cdot; \mu_k^*, S_k^* = \nu^{-1} (\nu-2) \Sigma_k^*, \nu\right).
\end{align*}
For Student distributions, the scale matrix $S_k$ was adapted so that the variance  {in  class $k$} was   $\Sigma_k$ and the number of degrees of freedom of each cluster was set to $\nu = 3$. 
The simulations dedicated to variance estimation were carried with a null mean vector and the covariance matrix $\Sigma_0$ given in Equation \eqref{eq:Sigma} in Appendix \ref{app:simDesign}.

\paragraph{Contamination scenarios.}
A contamination rate $\delta$ ranging from 0 (no contamination) to $50\%$ was applied to each cluster. Namely, a same fraction $\delta$ of the observations of each cluster $k = 1, \dots K$ was drawn with one of the five following contaminating distributions (hereafter referred to as 'scenarios')
\begin{enumerate}[($a$)]
  \item uniform distribution over the hypercube: $\mathcal{U}\{[-20, 20]^p\}$;
  \item Student distribution with null location vector, identity scale matrix and degree of freedom 1: $\mathcal{T}(0_p, I_p, 1)$; 
  \item Student distribution with location vector $\mu_k^*$, identity scale matrix and degree of freedom 1: $\mathcal{T}(\mu_k^*, I_p, 1)$; 
  \item Student distribution with null location vector, identity scale matrix and degree of freedom 2: $\mathcal{T}(0_p, I_p, 2)$; 
  \item Student distribution with location vector $\mu_k^*$, identity scale matrix and degree of freedom 2: $\mathcal{T}(\mu_k^*, I_p, 2)$. 
\end{enumerate}
Observe that, when considering one single cluster with null location parameter, scenarios ($c$) and ($e$) are equivalent to scenarios ($b$) and ($d$), respectively.
The contaminating distribution has no first two moments under scenarios ($b$) and ($c$), and no variance under scenarios ($d$) and ($e$). Under scenarios ($c$) and ($e$), the contaminating distribution has the same center as the corresponding cluster so the outliers can be considered as belonging to the cluster, whereas outliers arising from different clusters can not be distinguished under scenarios ($a$), ($b$) and ($d$).

\paragraph{Maximum likelihood estimates.}
For both Gaussian and Student mixtures models, we compared our results with the maximum likelihood estimates (MLE) provided by the R packages {\tt mclust} \cite{SFM16} for Gaussian mixtures and by the {\tt teigen} \cite{AWB18} for the Student mixtures. In the sequel, the corresponding algorithms and results will be referred to as GMM and TMM, respectively. The robust counterparts we propose will be referred to as RGMM and RTMM. For all the fours methods we carried the inference either with a fixed number of clusters $K$, or letting a model selection criterion (see below) choose an optimal number of clusters $\widehat{K}$.

\paragraph{Evaluation criteria.}
For each simulated dataset, we run the four algorithms (with fixed or selected $K$) and obtained estimates of the parameters $\mu_k$ and $\Sigma_k$, as well as a classification of each observation.
\begin{description}
 \item[Classification:] we used the adjusted Rand index (ARI) to compare the estimated classification with the simulated one. 
 \item[Parameter estimates:] when considering the true number of cluster $K$, we computed 
 \begin{itemize}
 \item the mean squared error for the center: $MSE(\mu) = K^{-1} \sum_k \|\mu_k^* - \widehat{\mu}_k\|^2/p$,
 \item {the mean squared error for the covariance: $MSE(\Sigma) = K^{-1} \sum_k \|\Sigma_k^* - \widehat{\Sigma}_k\|^2/p^2$.}
 \end{itemize}
\item[Model selection:] when considering the case of unknown number of cluster, we considered both the BIC and the ICL  criteria given in Equation \eqref{eq:modelSel}. 
\end{description}

\subsection{Variance and median estimation} \label{sec:simVariance}

The first part our simulations focuses on the robust estimation of the first two moments in one single cluster (no mixture). 

\subsubsection{Gaussian case}

\paragraph{No outlier.}
In this section,  we first consider the estimation of the variance and median in absence of outliers.  {To this aim, we consider} $X \sim \mathcal{N}(0,\Sigma)$, with $\Sigma = \Sigma_0$, as given in Equation \eqref{eq:Sigma}, Appendix \ref{app:simDesign}. We first focus on the accuracy of each method to estimate the variance. To do so, we consider $n= 10^5$ i.i.d copies of $X$ and estimate the MCM with the help of the Weiszfeld's algorithm. 

In Figure \ref{fig:estim_robust_cov}, we show the evolution of the quadratic mean error of the estimates with respect to the sample size. More precisely, we compared the estimates obtained with fix point algorithm, with $10$, $20$ and $50$ iterations, with the iterative gradient algorithm with  $10$, $20$ and $50$ iterations  and the averaged Robbins-Monro estimates (Robbins-Monro). \\
We also compared the behavior of the methods but with  {fixed} computation {budget}. More precisely, if a sample of size $N$ has been generated for the Monte-Carlo method for an iterative method with $T=50$, a sample of size $5 N$ is generated for an iterative method with $T=10$ iterations, and a sample of size $50N$ will be generated for the Robbins-Monro method.  {The results are} based on $\SR{}{B=}50$  {replicates}. 

We observe that all methods achieve convergence and have similar behaviors when they use samples with same sizes. Nevertheless, for  {fixed} computation  {budget}, the method based on the Robbins-Monro algorithm seams (without surprise) to lead to better results.

\begin{figure}[H]
\centering
    \includegraphics[scale=0.5]{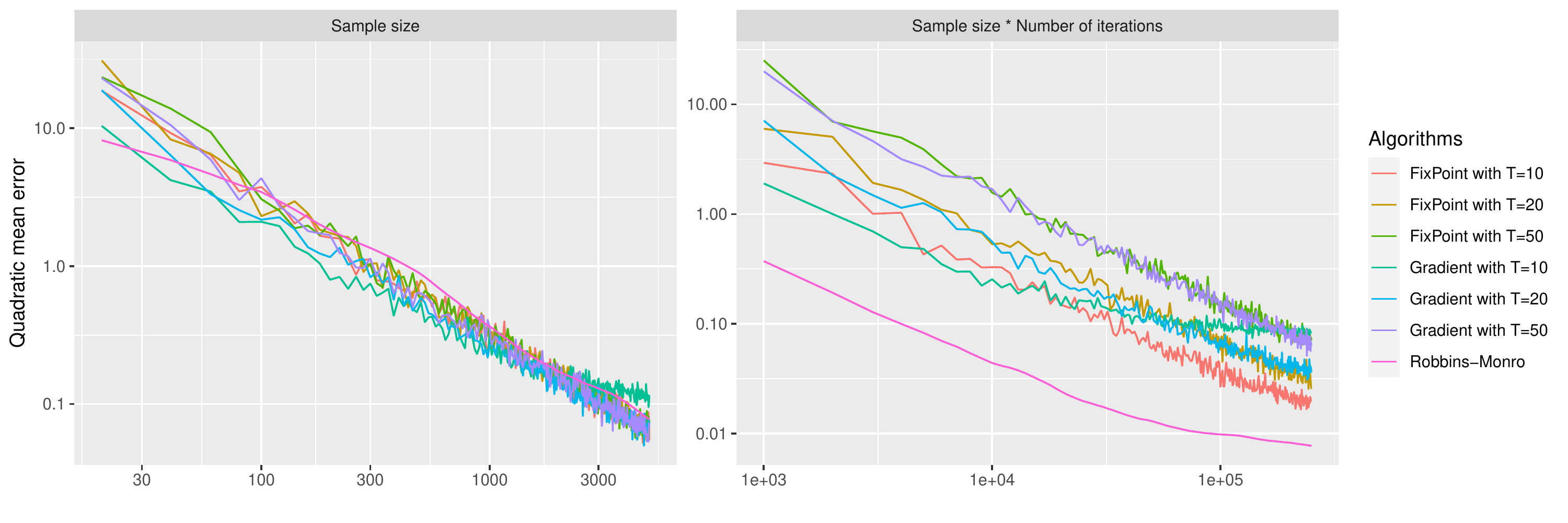}    
    \caption{Evolution of the quadratic mean error of the different methods with respect to the sample size (on the left) and to computation time (on the right).}
    \label{fig:estim_robust_cov}
\end{figure}

\paragraph{With outliers.}
We then introduced an increasing fraction $\delta$ of outliers according to scenarios ($a$), ($b$) and ($d$) described in Section \ref{sec:simDesign}. 
We considered samples with size $n = 5000$, and estimated the MCM with the help of the Weiszfeld algorithm (indicated by (W)) or with the ASGD (indicated by (R)). We estimated the  eigenvalues of the variance with the three proposed methods and with a sample size of $N = 2000$ for the Monte-Carlo method before building the variance. For iterative methods, we used $T=50$ iterations. 

 {
All robust methods provide accurate estimates of the variance, even in presence of a large fraction of outliers.}
In addition, one can see that even if Robbins-Monro method  {is slightly less precise than} the other robust alternatives, but performs well any way. 
 {Yet, as the} Robbins-Monro procedure is less expansive in term of {computation} time, and {because it turns out to be more accurate than the other methods with a same computational budget} (see Appendix \ref{app::variance}), this procedure will be preferred for robust mixture models.

\begin{table}[H]
\centering
\begin{tabular}[b]{cc|rrrrrrr}
& \rotatebox[origin=r]{360}{$\delta$ ($\%$)}  & \rotatebox{270}{FixPoint (R)}    & \rotatebox{270}{FixPoint (W)}    & \rotatebox{270}{Gradient (R)}    & \rotatebox{270}{Gradient (W)}    & \rotatebox{270}{Robbins (R)}    & \rotatebox{270}{Robbins (W)}    & \rotatebox[origin=l]{270}{Variance}       \\  
   \hline
\multirow{8}{*}{\rotatebox{90}{ {($a$): $U$}}}& 0  & 0.32 & 0.24 & 0.34 & 0.31 & 0.45 & 0.36 & \textbf{0.11} \\ 
&   2 & 0.39 & \textbf{0.34} & 0.36 & \textbf{0.34} & 0.40 & 0.36 & 39.75 \\ 
&   3 & \textbf{0.36} & 0.39 & 0.39 & \textbf{0.36} & 0.43 & 0.38 & 78.20 \\ 
&   5 & 0.63 & \textbf{0.51} & 0.59 & 0.57 & 0.57 & 0.59 & 212.60 \\ 
&   9 & 1.35 & 1.36 & 1.29 & 1.21 & 1.28 & \textbf{1.06} & 682.80 \\ 
&   16 & 4.01 & 3.88 & 3.91 & 3.89 & 3.41 & \textbf{3.36}   & $2.10^{3}$ \\ 
&  28 & 16.65 & 17.56 & 16.21 & 16.13 & 13.78 & \textbf{13.51} & $7.10^{3}$ \\ 
 & 50 & 154.52 & 165.05 & 133.19 & 142.32 &\textbf{ 109.12} & 116.59 & $2.10^{4}$ \\ 
\hline
\multirow{8}{*}{\rotatebox{90}{ {($b$): $T_1$}}} &   
0 & 0.31 & 0.29 & 0.32 & 0.34 & 0.38 & 0.40 & \textbf{0.10} \\ 
&   2 & 0.33 & 0.31 & \textbf{0.30} & 0.31 & 0.44 & 0.37 & $2.10^{8}$ \\ 
&   3 & 0.36 & \textbf{0.28} & 0.29 & 0.35 & 0.40 & 0.36 & $2.10^{7}$ \\ 
&   5 &\textbf{ 0.35} & 0.36 & 0.41 & 0.40 & 0.43 & 0.54 & $10^{9}$ \\ 
&   9 & 0.49 & \textbf{0.46} & 0.48 & 0.47 & 0.67 & 0.65 & $7.10^{9}$ \\ 
&   16 & 0.86 & 0.77 & 0.80 & \textbf{0.76 }& 0.98 & 0.93 & $8.10^{13}$ \\ 
&   28 & 1.74 & 1.76 & \textbf{1.64} & 1.78 & 2.01 & 1.92 & $5.10^{11}$ \\ 
&   50 & 5.49 & \textbf{5.28} & 5.38 & 5.52 & 5.59 & 5.84 & $2.10^{13}$ \\ 
   \hline
\multirow{8}{*}{\rotatebox{90}{ {($e$): $T_2$}}} &    0 & 0.29 & 0.28 & 0.37 & 0.29 & 0.46 & 0.33 & \textbf{0.12} \\ 
&   2 & 0.33 & 0.33 & \textbf{0.31} & 0.34 & 0.41 & 0.48 & 1.06 \\ 
&   3 & \textbf{0.35} & 0.40 & 0.42 & 0.38 & 0.63 & 0.41 & 0.59 \\ 
&   5 & 0.52 & 0.60 & \textbf{0.48} & 0.49 & 0.66 & 0.76 & 7.03 \\ 
&   9 & 0.86 & 1.02 & \textbf{0.79} & 0.98 & 1.10 & 1.20 & 6.10 \\ 
&   16 & \textbf{1.99} & 2.07 & 2.08 & 2.21 & 2.50 & 2.54 & 330.59  \\ 
&   28 & 5.80 & 5.59 & \textbf{5.50} & 5.88 & 5.92 & 6.20 & $9.10^{6 }$ \\ 
&   50 & \textbf{14.84} & 15.12 & 14.99 & 15.16 & 15.38 & 15.31 & $2.10^{4}$\\ 
   \hline
\end{tabular}
\caption{
 {Multivariate Gaussian case: Mean quadratic error} of the estimates of the variance for the different methods and for different contamination scenarios and fractions $\delta$.}
\end{table}

\subsubsection{Student case}

{We used a similar scheme for the Student distribution.}

\paragraph{No outlier.}

We considered a Student distribution with null mean vector, with variance $\Sigma_0$ and 3 degrees of freedom. 

We first focus on the accuracy of each method to estimate the variance and follow the same simulation plan as for the Gaussian case. Observe that in this case, the weighted averaged Robbins-Monro method is slightly less accurate for fixed sample sizes, but is slightly better for fixed computational budget. In Table \ref{tab::rob::student}, we first remark that the usual estimate of the variance clearly underperform, even for uncontaminated data. In addition, although gradient method with $50$ iterations is undoubtedly better, the Robbins-Monro alternative is a serious competitor.  Then, coupled with what has been observed in the Gaussian case, the method based on the Robbins-Monro algorithm seems the best option for estimating the variances of the clusters of robust mixture models.

\begin{figure}[H]
\centering
    \includegraphics[scale=0.5]{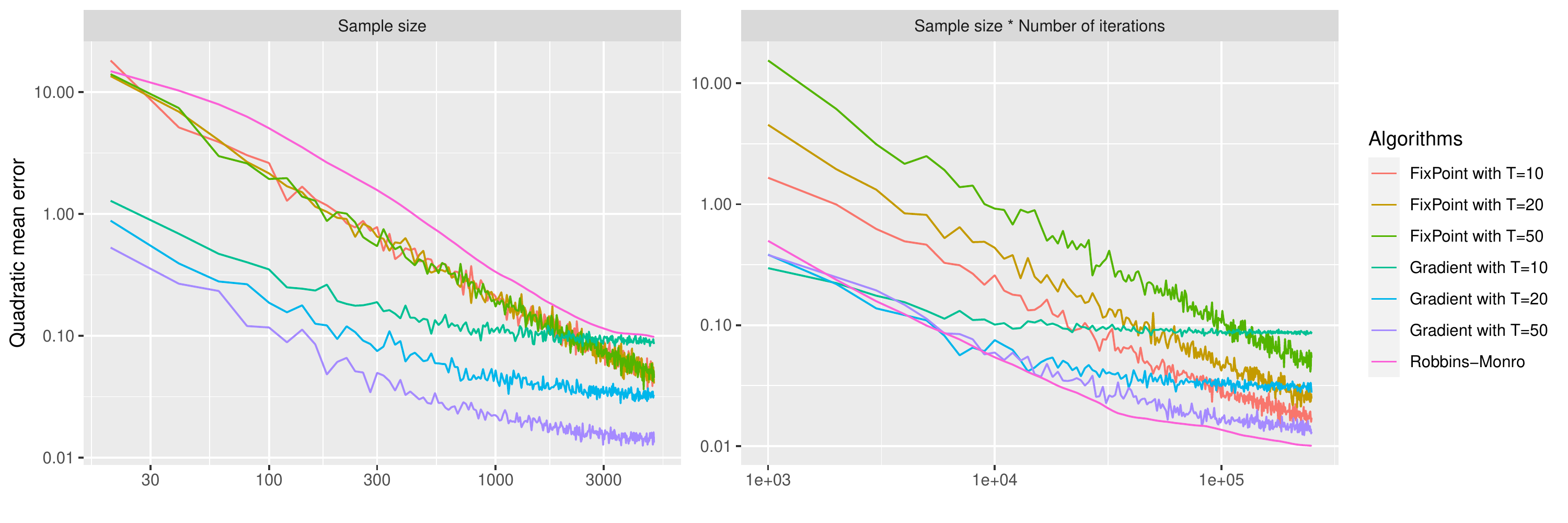}     
    \caption{Evolution of the quadratic mean error of the different methods with respect to the sample size (on the left) and to computation time (to the right).}
    \label{fig:Testim_robust_cov}
\end{figure}

\paragraph{With outliers.}

We then introduced an increasing fraction $\delta$ of outliers according to same three scenarios ($a$), ($b$) and ($d$) from Section \ref{sec:simDesign}. 

{The conclusion are the same as in the Gausian case.}

\begin{table}[H]
\centering
\begin{tabular}{cc|rrrrrrr}
& \rotatebox[origin=r]{360}{$\delta$ ($\%$)}  & \rotatebox{270}{FixPoint (R)}    & \rotatebox{270}{FixPoint (W)}    & \rotatebox{270}{Gradient (R)}    & \rotatebox{270}{Gradient (W)}    & \rotatebox{270}{Robbins (R)}    & \rotatebox{270}{Robbins (W)}    & \rotatebox[origin=l]{270}{Variance}       \\  
   \hline
\multirow{8}{*}{\rotatebox{90}{{($a$): $U$}}} &   0 & 0.29 & 0.25 & \textbf{0.20} & \textbf{0.20} & 0.50 & 0.46 & 19.78 \\ 
&   2 & 0.37 & 0.36 & 0.27 & \textbf{0.22} & 0.46 & 0.44 & 43.43\\ 
&   3 & 0.41 & 0.37 & 0.34 & \textbf{0.27} & 0.68 & 0.55 & 103.51\\ 
&   5 & 0.79 & 0.63 & 0.62 & \textbf{0.53} & 1.06 & 0.84 & 207.81 \\ 
&   9 & 2.01 & 1.82 & 1.91 &\textbf{ 1.63} & 2.34 & 1.90 & 733.99 \\ 
&   16 & 6.73 & 5.83 & 6.11 & \textbf{5.61} & 6.87 & 6.89 & $2.10^{3}$ \\ 
&   28 & 29.82 & 27.17 & 26.84 & \textbf{25.25} & 30.72 & 28.27 & $7.10^{3}$  \\ 
&   50 & 393.82 & 374.55 & 273.07 & \textbf{260.37} & 336.28 & 324.98 & $2.10^{4}$ \\ 
\hline
\multirow{8}{*}{\rotatebox{90}{{($b$): $T_1$}}} &   0 & 0.27 & 0.26 & \textbf{0.17} & 0.16 & 0.38 & 0.48 & 16.14 \\ 
&   2 & 0.37 & 0.31 & 0.21 & \textbf{0.17} & 0.52 & 0.46 & $10^{8}$\\ 
&   3 & 0.35 & 0.27 & 0.23 & \textbf{0.20} & 0.52 & 0.45 & $10^{10}$\\ 
&   5 & 0.44 & 0.39 & 0.31 & \textbf{0.27} & 0.62 & 0.69 & $3.10^{9}$ \\ 
&   9 & 0.83 & 0.75 & 0.67 & \textbf{0.59} & 1.22 & 0.93 & $10^{10}$ \\ 
&   16 & 2.18 & 1.97 & 1.90 & \textbf{1.77} & 2.74 & 1.98 & $2.10^{10}$ \\ 
&   28 & 6.54 & 6.17 & 6.08 & \textbf{5.64} & 7.00 & 6.05 & $5.10^{12}$  \\ 
&   50 & 32.39 & 30.08 & 29.16 & \textbf{27.99} & 31.48 & 29.79 & $2.10^{18}$ \\ 
\hline
\multirow{8}{*}{\rotatebox{90}{{($e$): $T_2$}}} &   0 & 0.30 & 0.26 & 0.19 & \textbf{0.18} & 0.40 & 0.34 & 12.77 \\ 
&   2 & 0.37 & 0.30 & \textbf{0.21} & \textbf{0.21} & 0.51 & 0.45 & 3.81  \\ 
&   3 & 0.31 & 0.32 &\textbf{ 0.21} & \textbf{0.21} & 0.42 & 0.40 & 9.72 \\ 
 &  5 & 0.33 & 0.29 & \textbf{0.22} & \textbf{0.22} & 0.50 & 0.43 & 38.64  \\ 
 &  9 & 0.44 & 0.39 & 0.34 & \textbf{0.30} & 0.61 & 0.59 & 14.00 \\ 
 &  16 & 0.84 & 0.80 & \textbf{0.69} & \textbf{0.69} & 0.98 & 0.95 & 778.37 \\ 
 &  28 & 2.08 & 1.96 & 1.95 & \textbf{1.91} & 2.27 & 2.14 & $3.10^{3}$ \\ 
 &  50 & 6.57 & 6.35 & 6.56 & \textbf{6.42} & 7.23 & 6.45 & 401.01 \\ 
   \hline
\end{tabular}
\caption{\label{tab::rob::student} {Multivariate Student case: Mean quadratic error} of the estimates of the variance for the different methods and for different contamination scenarios and fractions $\delta$.}
\end{table}

\subsection{Mixture models} \label{sec:simMixture}

The second part our simulation deals with mixture models. 

\subsubsection{Gaussian mixture model.}
We simulated $B=100$ datasets according to a Gaussian mixture model with each of the parameter configurations described in Section \ref{sec:simDesign}. We only present here the results for a total sample size of $n = 1500$) (that is $n_k = 500$ observations in each group).  We did not observe substantial differences between the results obtained when selecting the number of clusters $K$ with $BIC$ and $ICL$. As a consequence, we only present the results obtained with $BIC$. 

The first two columns of Figure \ref{fig:simGMM1500} compare the results of maximum-likelihood (GMM) inference with the proposed approach (RGMM) in terms of classification. When fixing the number of clusters to its true value $K^* = 3$, we observe a dramatic drop of the classification accuracy of GMM estimation, even for a very moderate fraction of outliers ($\delta=2\%$), as compared to RGMM, in all scenarios. We observe that estimating the number of clusters with $BIC$ improves the classification performances of GMM, at the price of an increase of the number of clusters. On the contrary, the RGMM approach keeps selecting the right number of clusters, even with a medium fraction of outliers $(\delta \sim 10-20\%)$. As a consequence, model selection does not  {significantly} improve the classification accuracy of RGMM. Lastly, we observe that the difference between GMM and RGMM is even more obvious when outliers can each be associated with one cluster, that is under scenarios ($c$) and ($e$), as opposed to scenarios ($b$) and ($d$), respectively. 

\begin{figure}[ht]
  \centering
    \begin{tabular}{c|m{.2\textwidth}m{.2\textwidth}|m{.2\textwidth}m{.2\textwidth}}
      & 
      \multicolumn{2}{c|}{Classification} & 
      \multicolumn{2}{c}{Parameter estimation} \\ 
      &
      \multicolumn{1}{c}{ARI} & \multicolumn{1}{c|}{$\widehat{K}$} & 
      \multicolumn{1}{c}{$MSE(\mu)$} & \multicolumn{1}{c}{$MSE(\Sigma)$} \\ 
      \hline
      ($a$) &
      \includegraphics[width=.2\textwidth, trim=10 10 10 50, clip=]{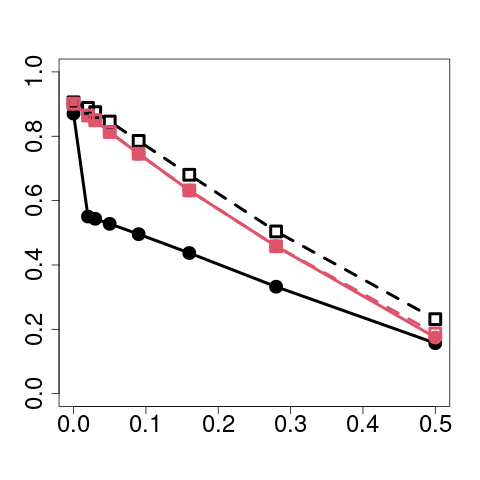} &
      \includegraphics[width=.2\textwidth, trim=10 10 10 50, clip=]{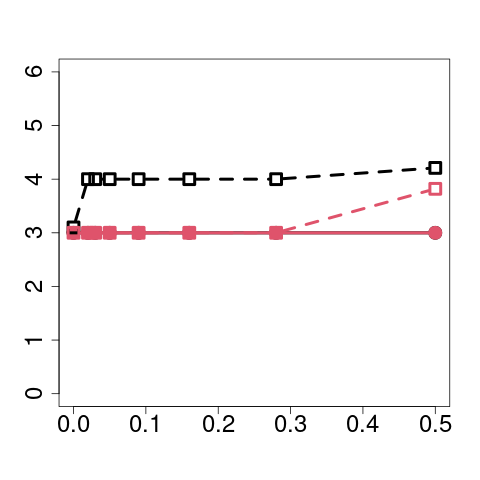} &
      \includegraphics[width=.2\textwidth, trim=10 10 10 50, clip=]{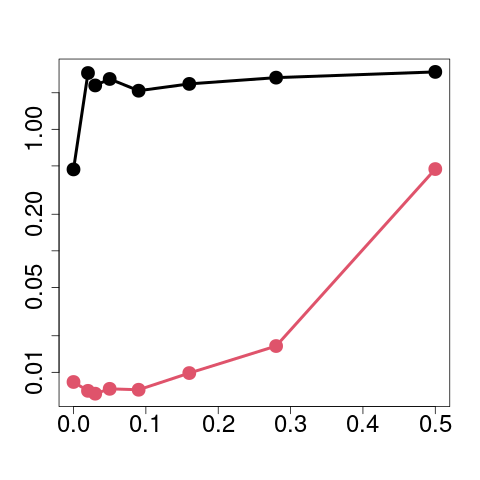} &
      \includegraphics[width=.2\textwidth, trim=10 10 10 50, clip=]{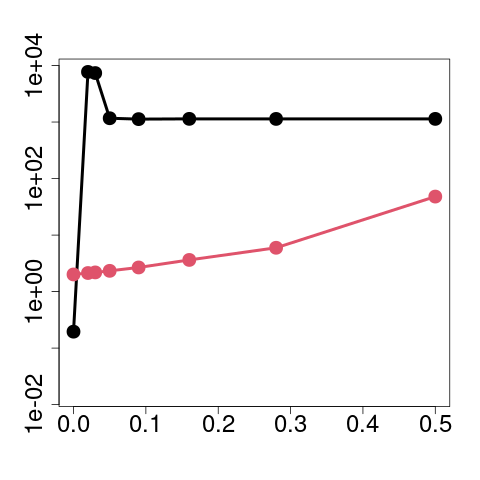} \\
      ($b$) &
      \includegraphics[width=.2\textwidth, trim=10 10 10 50, clip=]{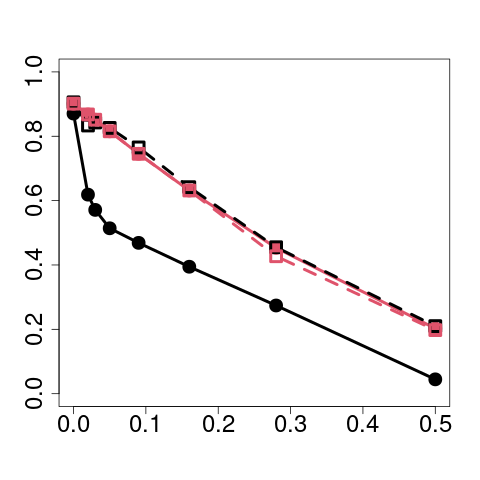} &
      \includegraphics[width=.2\textwidth, trim=10 10 10 50, clip=]{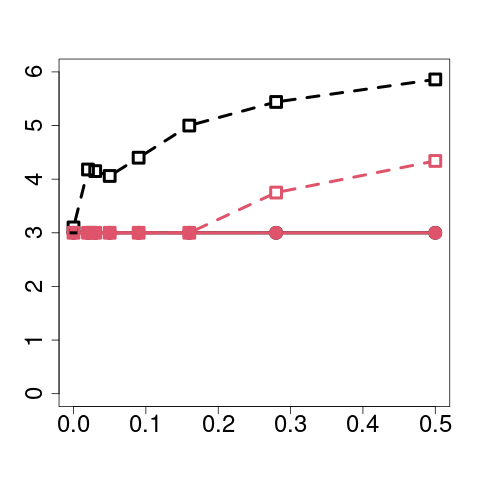} &
      \includegraphics[width=.2\textwidth, trim=10 10 10 50, clip=]{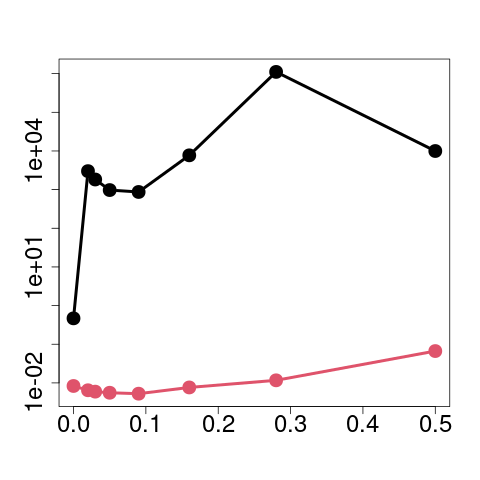} &
      \includegraphics[width=.2\textwidth, trim=10 10 10 50, clip=]{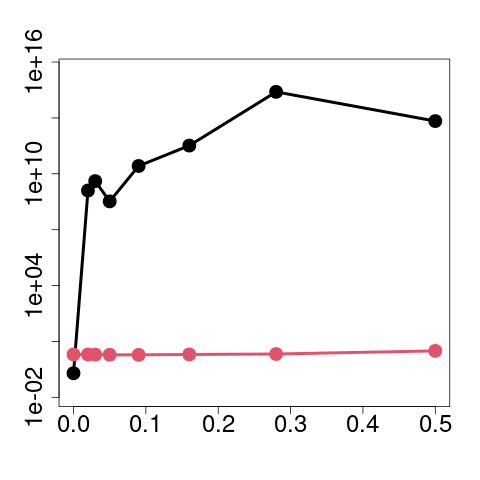} \\
      ($c$) &
      \includegraphics[width=.2\textwidth, trim=10 10 10 50, clip=]{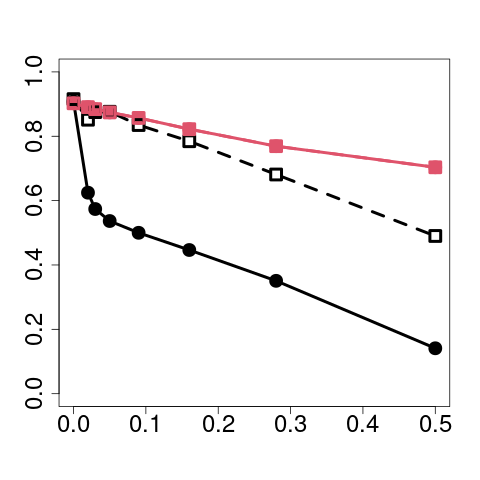} &
      \includegraphics[width=.2\textwidth, trim=10 10 10 50, clip=]{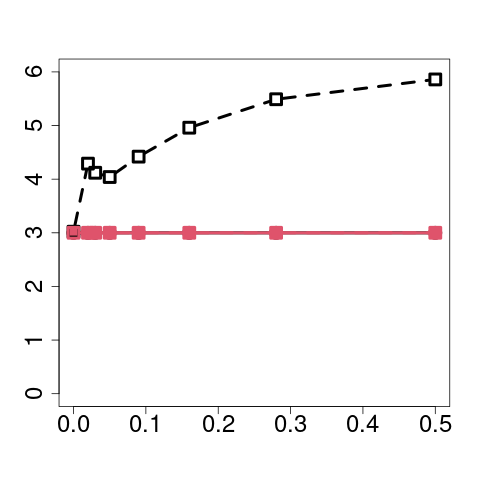} &
      \includegraphics[width=.2\textwidth, trim=10 10 10 50, clip=]{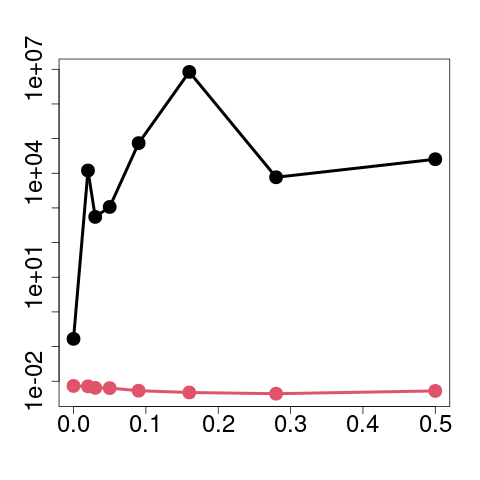} &
      \includegraphics[width=.2\textwidth, trim=10 10 10 50, clip=]{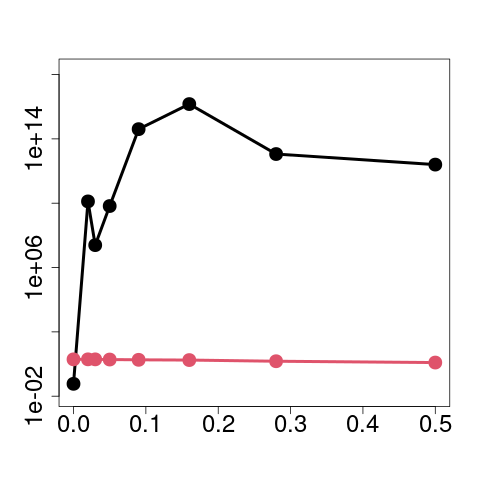} \\
      ($d$) &
      \includegraphics[width=.2\textwidth, trim=10 10 10 50, clip=]{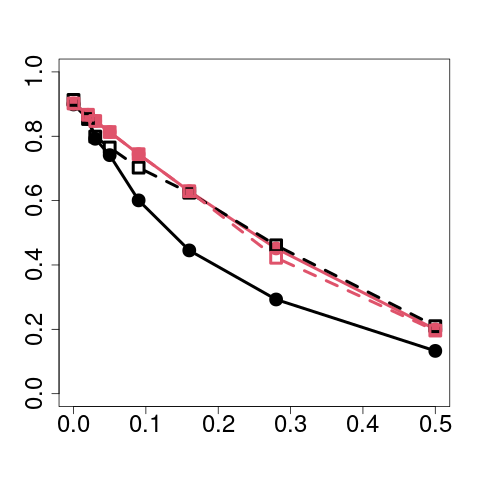} &
      \includegraphics[width=.2\textwidth, trim=10 10 10 50, clip=]{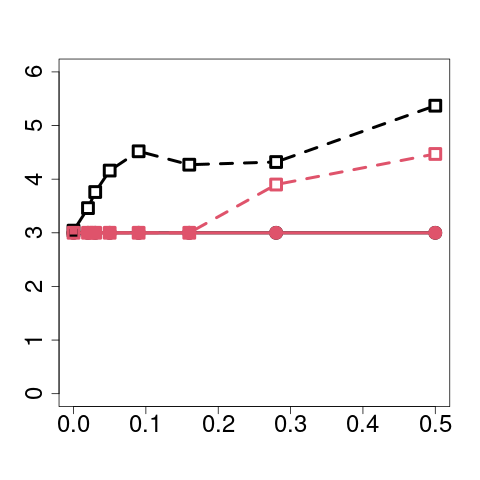} &
      \includegraphics[width=.2\textwidth, trim=10 10 10 50, clip=]{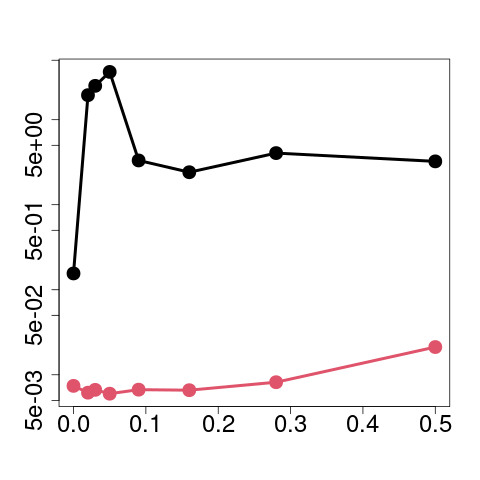} &
      \includegraphics[width=.2\textwidth, trim=10 10 10 50, clip=]{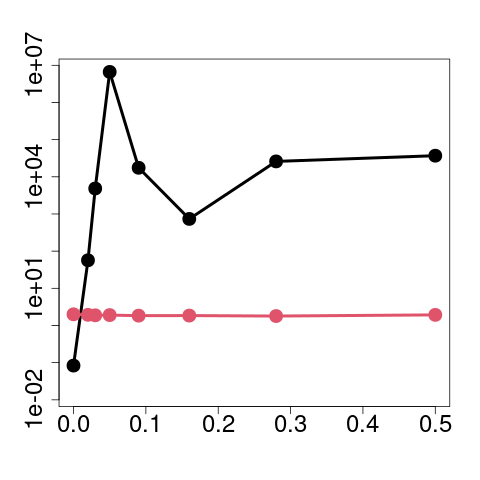} \\
      ($e$) &
      \includegraphics[width=.2\textwidth, trim=10 10 10 50, clip=]{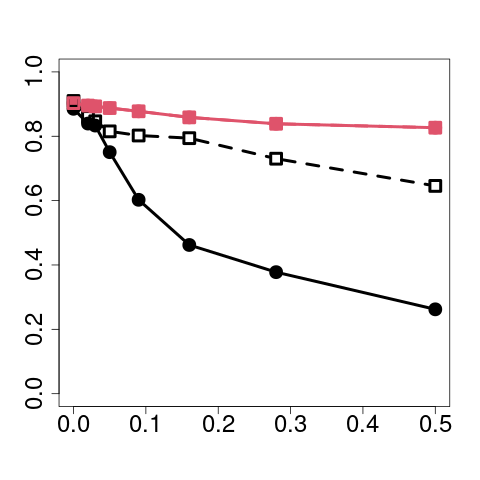} &
      \includegraphics[width=.2\textwidth, trim=10 10 10 50, clip=]{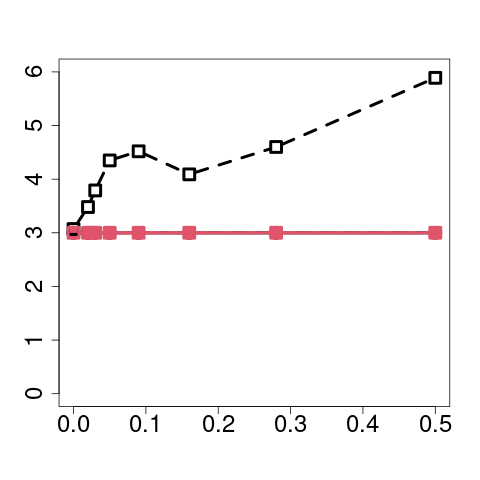} &
      \includegraphics[width=.2\textwidth, trim=10 10 10 50, clip=]{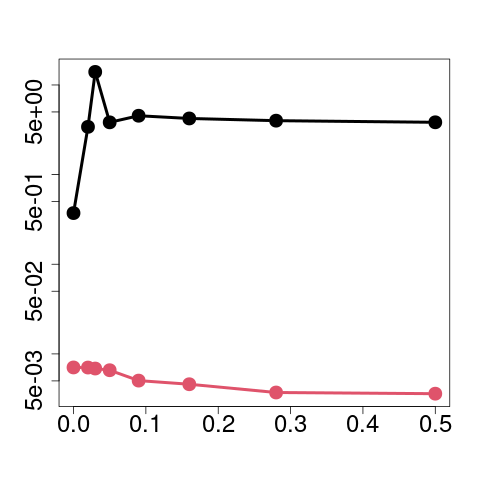} &
      \includegraphics[width=.2\textwidth, trim=10 10 10 50, clip=]{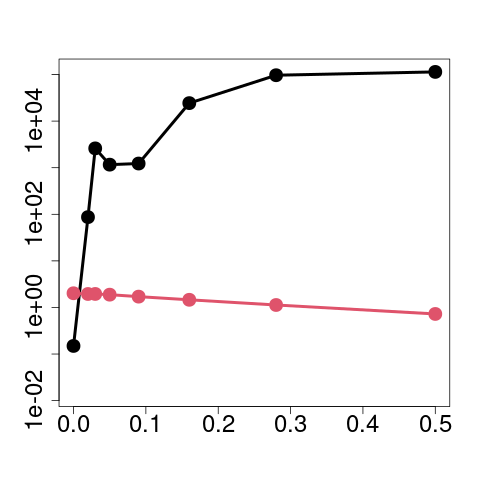} 
    \end{tabular}
    \caption{Gaussian mixture model: classification accuracy ($ARI$), estimated number of clusters $\widehat{K}$, estimation error fu the mean ($MSE(\mu)$) and for the variance ($MSE(\Sigma)$) for scenarios ($a$) to ($e$), with $n_k = 500$ observation in each of the $K^*$ clusters ($n = 1500$). Black: maximum likelihood (GMM); red: robust estimation (RGMM). Solid line ($\bullet$): with true number of clusters $K^*$; dotted line ($\square$): with number of clusters estimated with $BIC$. \label{fig:simGMM1500}}
\end{figure}

The last two columns of Figure \ref{fig:simGMM1500} compare the respective accuracies of GMM and RGMM in terms of parameter estimation. The precision achieved by RGMM is several order of magnitude better than this of GMM, and, except under scenario ($a$), this accuracy remains the same for large contamination fractions (up to $\delta = 50\%$). Again, model selection does not improve the estimation precision of the robust approach. 


Figure \ref{fig:simGMM300}, given in Appendix \ref{app:simResults}, is the same as Figure \ref{fig:simGMM1500}, but was obtained with $n_k = 100$ observations in each cluster (that is $n = 300$). The same conclusions, although less contrasted, can be drawn from it.


\FloatBarrier

\subsubsection{Student mixture model.}
We then simulated $B=100$ datasets according to a Student mixture model with the same set of parameter configurations (see Section \ref{sec:simDesign}). Again, we only present here the results with $n_k = 500$ observations in each group ($n = 1500$).  For the same reason as in the Gaussian case, we only present the results obtained with the $BIC$ criterion. 

Figure \ref{fig:simTMM1500} is organized in the same way as Figure \ref{fig:simGMM1500}. In terms of classification, we observe a dramatic drop of the accuracy obtained with maximum likelihood inference (TMM: as performed by the {\tt teigen} R package), as compared to its robust counterpart (RTMM). We also observe that, depending on the simulation scenario, the classification accuracy of the robust approach decreases more or less rapidely, the better results being obtained under the scenarios where outliers can each be associated with a clusters (($c$) and ($e$)).

\begin{figure}[ht]
  \centering
    \begin{tabular}{c|m{.2\textwidth}m{.2\textwidth}|m{.2\textwidth}m{.2\textwidth}}
      & 
      \multicolumn{2}{c|}{Classification} & 
      \multicolumn{2}{c}{Parameter estimation} \\ 
      &
      \multicolumn{1}{c}{ARI} & \multicolumn{1}{c|}{$\widehat{K}$} & 
      \multicolumn{1}{c}{$MSE(\mu)$} & \multicolumn{1}{c}{$MSE(\Sigma)$} \\ 
      \hline
      ($a$) &
      \includegraphics[width=.2\textwidth, trim=10 10 10 50, clip=]{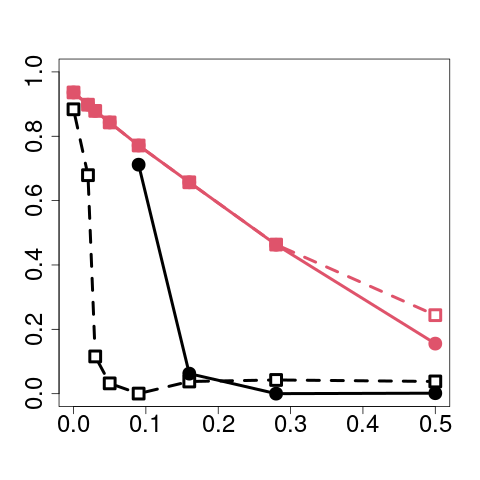} &
      \includegraphics[width=.2\textwidth, trim=10 10 10 50, clip=]{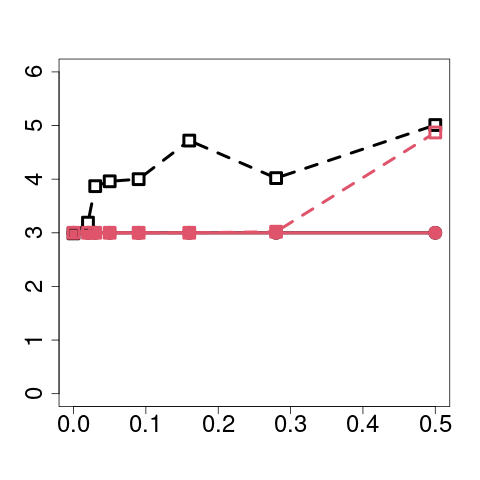} &
      \includegraphics[width=.2\textwidth, trim=10 10 10 50, clip=]{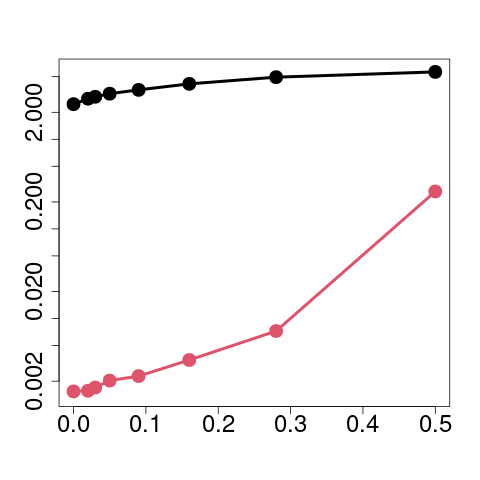} &
      \includegraphics[width=.2\textwidth, trim=10 10 10 50, clip=]{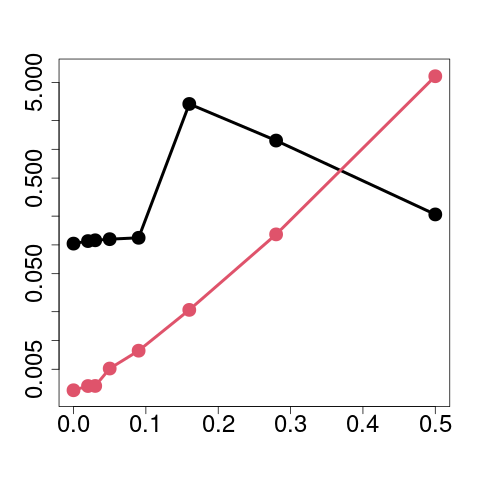} \\
      ($b$) &
      \includegraphics[width=.2\textwidth, trim=10 10 10 50, clip=]{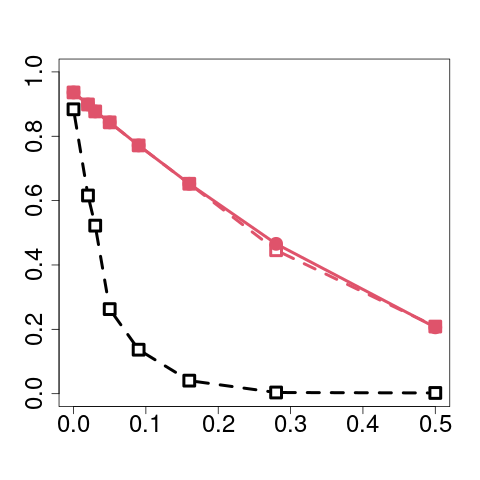} &
      \includegraphics[width=.2\textwidth, trim=10 10 10 50, clip=]{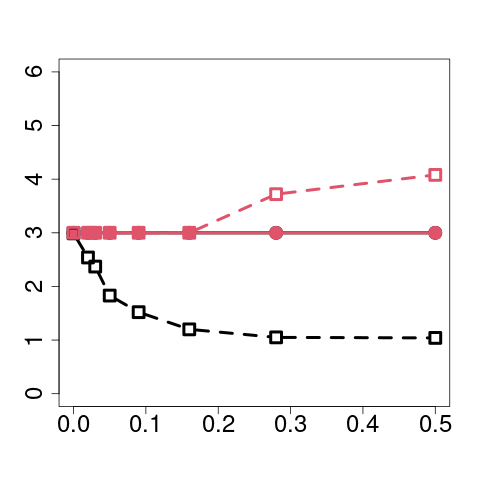} &
      \includegraphics[width=.2\textwidth, trim=10 10 10 50, clip=]{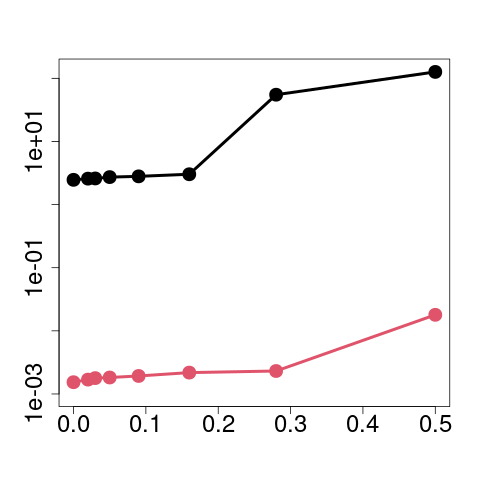} &
      \includegraphics[width=.2\textwidth, trim=10 10 10 50, clip=]{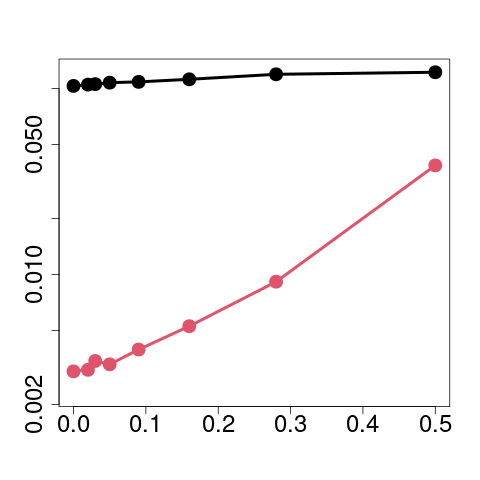} \\
      ($c$) &
      \includegraphics[width=.2\textwidth, trim=10 10 10 50, clip=]{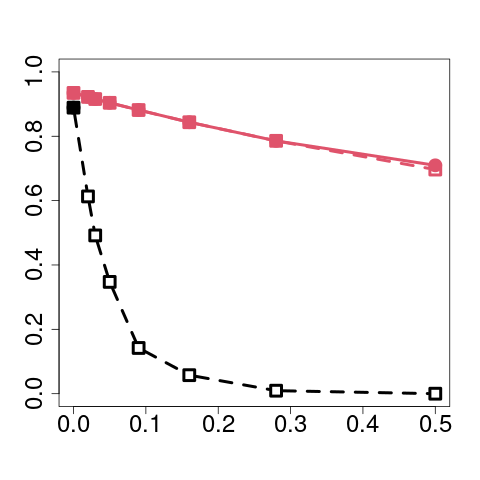} &
      \includegraphics[width=.2\textwidth, trim=10 10 10 50, clip=]{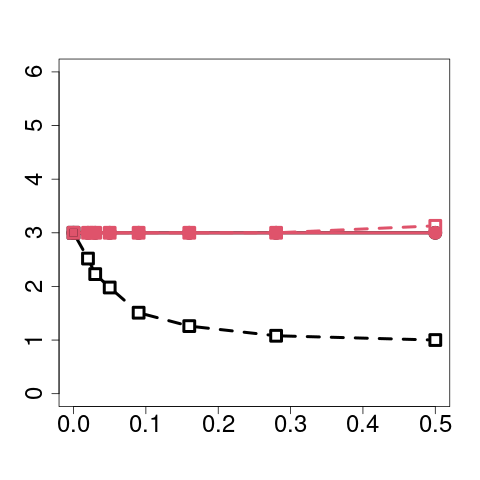} &
      \includegraphics[width=.2\textwidth, trim=10 10 10 50, clip=]{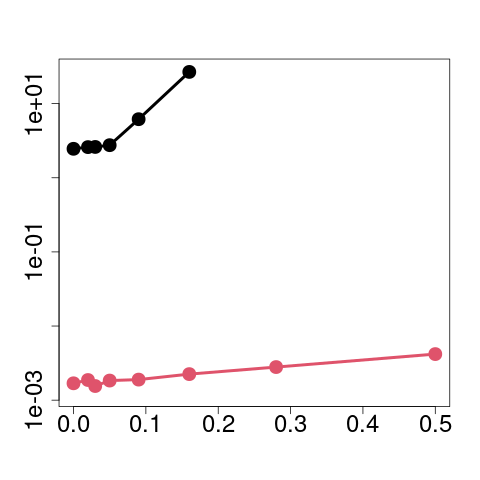} &
      \includegraphics[width=.2\textwidth, trim=10 10 10 50, clip=]{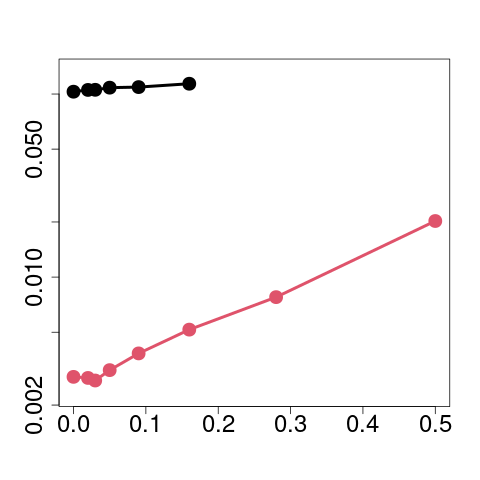} \\
      ($d$) &
      \includegraphics[width=.2\textwidth, trim=10 10 10 50, clip=]{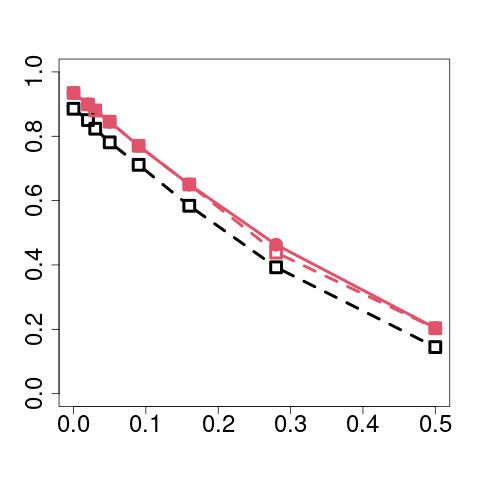} &
      \includegraphics[width=.2\textwidth, trim=10 10 10 50, clip=]{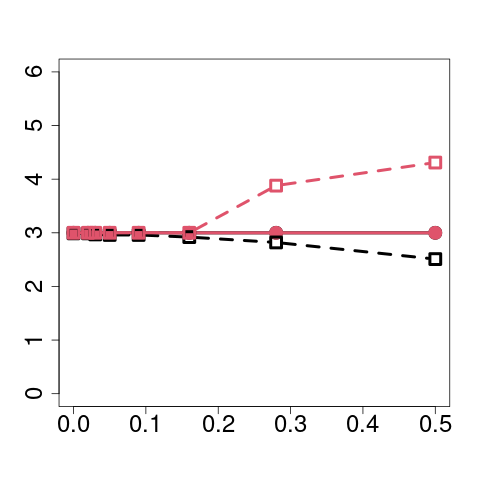} &
      \includegraphics[width=.2\textwidth, trim=10 10 10 50, clip=]{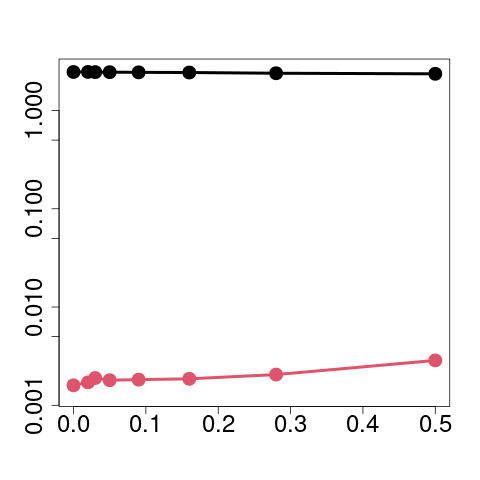} &
      \includegraphics[width=.2\textwidth, trim=10 10 10 50, clip=]{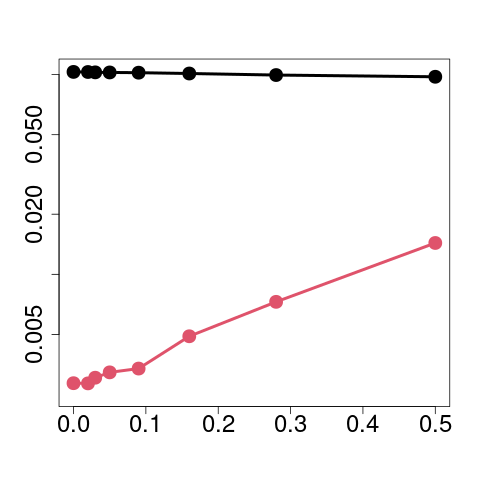} \\
      ($e$) &
      \includegraphics[width=.2\textwidth, trim=10 10 10 50, clip=]{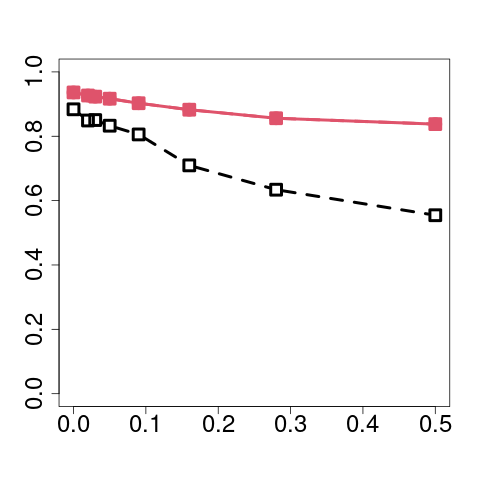} &
      \includegraphics[width=.2\textwidth, trim=10 10 10 50, clip=]{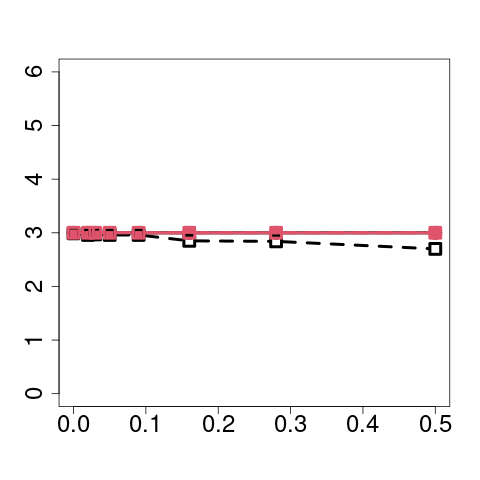} &
      \includegraphics[width=.2\textwidth, trim=10 10 10 50, clip=]{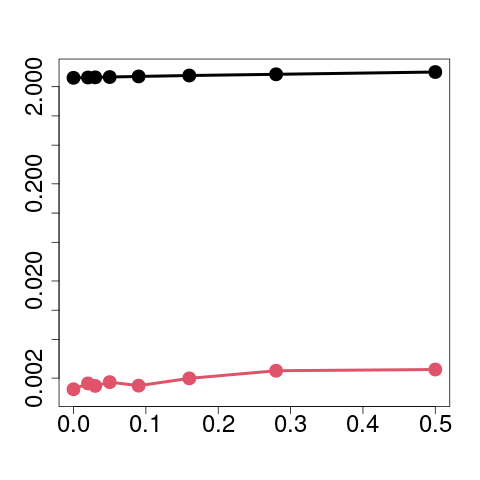} &
      \includegraphics[width=.2\textwidth, trim=10 10 10 50, clip=]{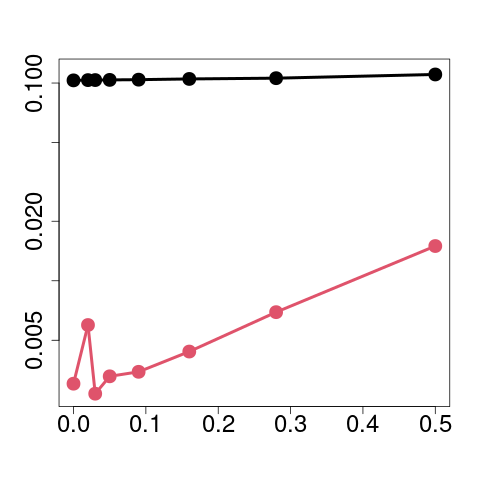} 
    \end{tabular}
    \caption{Student mixture model: classification accuracy ($ARI$), estimated number of clusters $\widehat{K}$, estimation error fu the mean ($MSE(\mu)$) and for the variance ($MSE(\Sigma)$) for scenarios ($a$) to ($e$), with $n_k = 500$ observation in each of the $K^*$ clusters ($n = 1500$). Same legend as Figure \ref{fig:simGMM1500}. \label{fig:simTMM1500}}
\end{figure}

The last two columns of Figure \ref{fig:simTMM1500} also shows better performances of the robust approach RTMM as compared to maximum likelihood TMM in terms of precision accuracy. Observe that several curves associated with TMM display an erratic behavior due to convergence issues of the EM algorithm (see Figure \ref{fig:simTMM1500fail} in Appendix \ref{app:simResults}).

Similarly to de Gaussian case, Figure \ref{fig:simTMM300}, given in Appendix \ref{app:simResults}, is the same as Figure \ref{fig:simTMM1500} for $n_k = 100$ observations per cluster ($n = 300$): again similar conclusions can be drawn from it.

\FloatBarrier

\paragraph{Aknowledgement.}
We are grateful to the INRAE MIGALE bioinformatics facility (MIGALE, INRAE, 2020. Migale bioinformatics Facility, \url{doi:10.15454/1.5572390655343293E12}) for providing computing and storage resources.

\section{Proofs}\label{sec:proofs}
\begin{proof}[Proof of Proposition \ref{prop1}]
Remark that one can rewrite 
\[
\mathbb{E}_{\theta^{*}} \left[ G_{2}(m)  \right] = \mathbb{E}_{\theta^{*}} \left[ \sum_{k=1}^{K} \mathbf{1}_{Z=k}  \left\| X - m_{k} \right\|   \right] 
\]
and this function is Frechet-differentiable with 
\[
\nabla_{m_{k}} \mathbb{E}_{\theta^{*}} \left[ G_{2}(m)  \right] = - \mathbf{E}_{\theta^{*}} \left[   \frac{X - m_{k}}{\left\| X - m_{k} \right\|} \mathbf{1}_{Z=k}   \right] .
\]
and the zero of the gradient of $G_{2}$ correspond to the median of the classes.
Since $X |Z =k$ is symmetric, one has $\nabla_{m_{k}}\mathbb{E}_{\theta^{*}} \left[G_{2} \left( \mu_{k}^{*} \right) \right] = 0$, i.e $\mu_{k}^{*} = m_{k}^{*}$. In a same way, one has
\begin{align*}
\nabla_{V_{k}} \mathbb{E}_{\theta^{*}} \left[ G_{3}(V) \right]  & = - \mathbb{E}_{\theta^{*}} \left[  \frac{\left( X - m_{k}^{*} \right)\left( X - m_{k}^{*} \right)^{T} - V}{\left\|   \left( X - m_{k}^{*} \right)\left( X - m_{k}^{*} \right)^{T} - V  \right\|_{F} }  \mathbf{1}_{Z = k} \right] 
\end{align*}
and the zero of the gradient correspond to the MCM of the class $V_{k}^{*}$. Then, since $X|Z =k$ is symmetric,  the zero of the gradient satisfies $\Sigma_{k}^{*} = \Psi \left( V_{k}^{*} \right)$.
\end{proof}

\begin{proof}[Proof of Proposition \ref{prop:fixpoint}]

Since $\pi^{*}$ is a zero of the gradient of the Lagrangian, one has $\pi_{k}^{*} = \mathbb{E}\left[ \tau_{k}(X) \right]$, i.e one has $\pi^{*} = g_{1}\left( \pi^{*} \right)$. In a same way, $m_{k}^{*}$ is a zero of $\nabla_{m_{k}} G_{2} $, where $\nabla_{m_{k}}$ denotes the partial gradient with respect to $m_{k}$. Furthermore,
\begin{align*}
\nabla_{m_{k}}G_{2}(m) & = - \mathbb{E}\left[ \tau_{k}(X)\frac{\left( X - m_{k} \right)}{\left\| X - m_{k} \right\|  } \right] = 0  & 
& \Leftrightarrow & 
\mathbb{E}\left[ \tau_{k}(X) \frac{X}{\left\| X - m_{k} \right\|} \right] 
& = m_{k} \mathbb{E}\left[ \tau_{k}(X) \frac{1}{\left\| X - m_{k} \right\|} \right] \\
& & & \Leftrightarrow & 
m_{k} & = g_{2,k}\left( m_{k} \right) ,
\end{align*}
and in a particular case, $m^{*}$ is a minimizer of $G_{2}$ if and only if $m^{*} = g_{2}\left( m^{*} \right)$.
In a same way, denoting by $\nabla_{V_{k}}$ the gradient of $G_{3}$ with respect to $V_{k}$, one has
\begin{align*}
0 = \nabla_{V_{k}}G_{3}(m^{*},V) & = - \mathbb{E}\left[ \tau_{k}(X) \frac{\left( X - m_{k}^{*} \right) \left( X - m_{k}^{*} \right) - V_{k}}{\left\| \left( X - m_{k}^{*} \right) \left( X - m_{k}^{*} \right) - V_{k} \right\|_{F} }\right]  \\
\Leftrightarrow \quad \mathbb{E}\left[ \tau_{k}(X) \frac{\left( X - m_{k}^{*} \right) \left( X - m_{k}^{*} \right) }{\left\| \left( X - m_{k}^{*} \right) \left( X - m_{k}^{*} \right) - V_{k} \right\|_{F}} \right] & = V_{k} \mathbb{E}\left[ \tau_{k}(X) \frac{1}{\left\| \left( X - m_{k}^{*} \right) \left( X - m_{k}^{*} \right) - V_{k} \right\|_{F}} \right] ,
\end{align*}
which concludes the proof.

\end{proof}

\bibliographystyle{apalike}
\bibliography{biblio}

\appendix
\section{Appendix}\label{sec:appendix}
\subsection{Simulation design} \label{app:simDesign}

The vectors of means used for the simulation study where the following:
\begin{align} \label{eq:mu}
  \mu_1^\intercal 
  & = [\begin{array}{rrrrr} 0 & 0 & 0 & 0 & 0 \end{array}], & 
  \mu_2^\intercal 
  & = [\begin{array}{rrrrr} 3 & 3 & 3 & 3 & -3 \end{array}], &
  \mu_3^\intercal 
  & = [\begin{array}{rrrrr} -3 & -3 & -3 & -3 & -3 \end{array}], 
\end{align}
and the variance matrices were
\begin{align} \label{eq:Sigma}
  \Sigma_1 & = \left[ 
      \begin{array}{rrrrr} 
      2 & 0.43 & 0.41 & 0.15 & 0.68 \\ 
      0.43 & 2 & 0.7 & 0.49 & 0.89 \\ 
      0.41 & 0.7 & 2 & 0.17 & 0.42 \\ 
      0.15 & 0.49 & 0.17 & 2 & 0.43 \\ 
      0.68 & 0.89 & 0.42 & 0.43 & 2 
      \end{array} 
    \right] &
  \Sigma_2 & = \left[ 
    \begin{array}{rrrrr} 
    1 & 0.46 & 0.17 & 0.04 & 1.06 \\ 
    0.46 & 2 & 0.61 & 0.18 & 1.22 \\ 
    0.17 & 0.61 & 3 & 0.7 & 0.65 \\ 
    0.04 & 0.18 & 0.7 & 4 & 0.16 \\ 
    1.06 & 1.22 & 0.65 & 0.16 & 5 
    \end{array} 
    \right] \nonumber \\
    \\
  \Sigma_3 & = \left[ 
    \begin{array}{rrrrr} 
    1 & 0.6 & 0.11 & 0.03 & 0.26 \\ 
    0.6 & 0.5 & 0.09 & 0.02 & 0.17 \\ 
    0.11 & 0.09 & 0.33 & 0.03 & 0.04 \\ 
    0.03 & 0.02 & 0.03 & 0.25 & 0.01 \\ 
    0.26 & 0.17 & 0.04 & 0.01 & 0.2 
    \end{array} 
    \right] &
  \Sigma_0 & = \left[ 
      \begin{array}{rrrrr} 
      4 & 0.86 & 0.83 & 0.29 & 1.35 \\ 
      0.86 & 4 & 1.4 & 0.97 & 1.79 \\ 
      0.83 & 1.4 & 4 & 0.35 & 0.84 \\ 
      0.29 & 0.97 & 0.35 & 4 & 0.86 \\ 
      1.35 & 1.79 & 0.84 & 0.86 & 4 
      \end{array}   
    \right] \nonumber 
\end{align}

\subsection{Illustration of the results with the R package RGMM}

In this Section, we explain how to use the R package \texttt{RGMM} to illustrate the results. First, let us consider the following R code:
\begin{verbatim}
> mu <- matrix( c(rep(0,10),rep(2,10),rep(-2,10)), byrow=T, nrow=3)
> ech <- Gen_MM(nk = rep(200,3), delta=0.1,mu=mu)
> X<- ech$X
> Result <- RobMM(X)
\end{verbatim}
The function \texttt{Gen$\_$MM} enables to generate a sample of mixture model (Gaussian, Student or Laplace), whose centers are the raws of the matrix \texttt{mu}. The number  of data by cluster is given by \texttt{nk} while \texttt{delta} gives the proportion of contaminated data. The function \texttt{RobMM} gives the results obtained with the help of our method.  One can see the vignette for more details and to see the different options. 

We now focus on the function \texttt{RMMplot} which enables to illustrate the results. More precisely, we now comment the different available graphics.

\begin{verbatim}
> RMMplot(Result,graph=c('Two_Dim'))
\end{verbatim}

\begin{figure}[H]
\centering
    \includegraphics[scale=0.5]{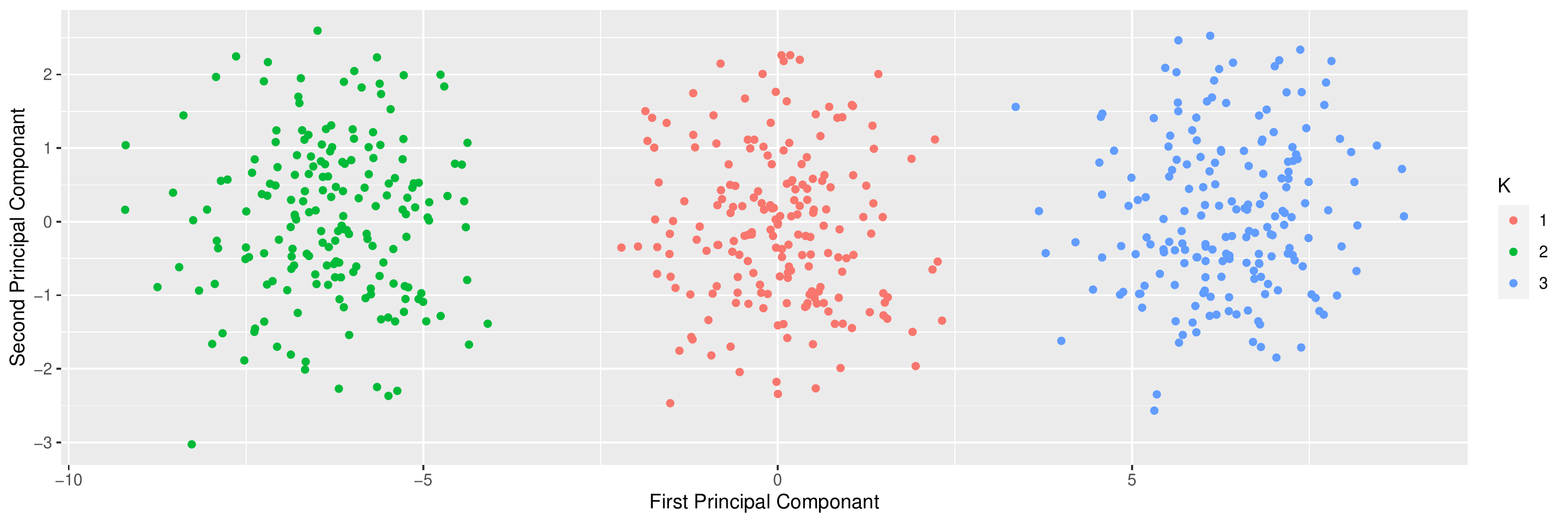}    
\end{figure}

The option \texttt{'Two$\_$Dim'} enables to represent the 2 first principal components of the data using robust principal component analysis components (RPCA) (see \cite{CG2015}).

\begin{verbatim}
> RMMplot(Result,graph=c('Two_Dim_Uncertainty'))
\end{verbatim}

\begin{figure}[H]
\centering
    \includegraphics[scale=0.5]{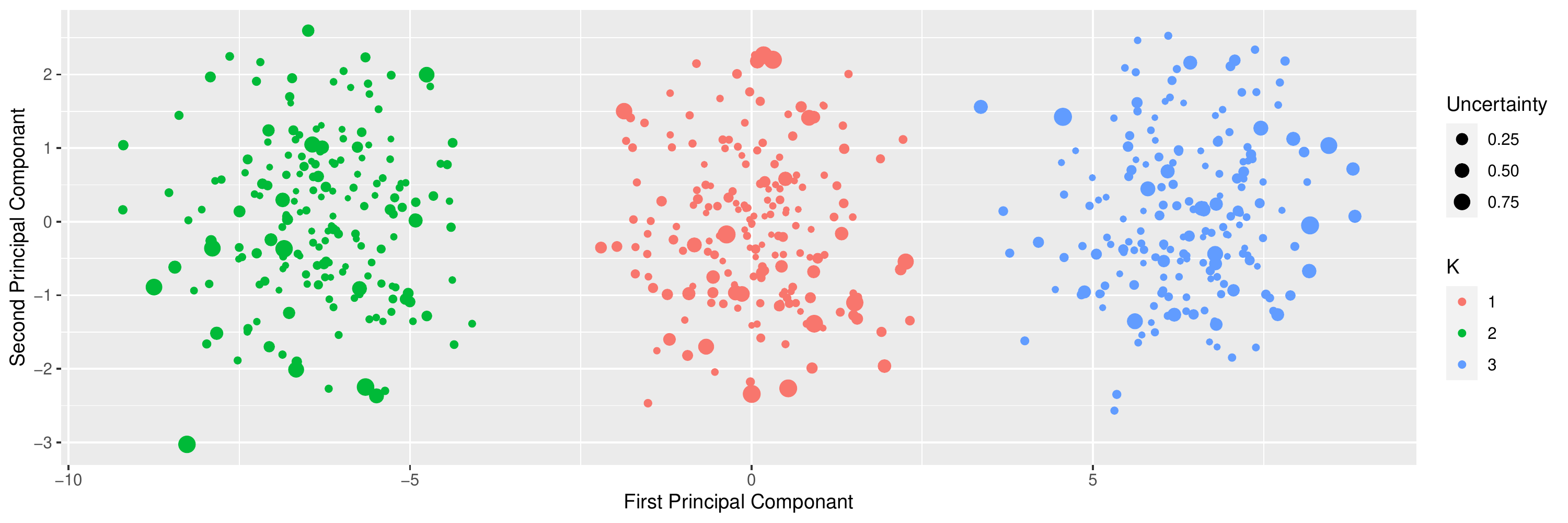}    
\end{figure}

The option \texttt{'Two$\_$Dim$\_$Uncertainty'} enables also to represent the 2 first principal components, but the size of the points is proportional to the uncertainty of the classification of the data.

\begin{verbatim}
> RMMplot(Result,graph=c('ICL'))
> RMMplot(Result,graph=c('BIC'))
\end{verbatim}

\begin{figure}[H]
\centering
    \includegraphics[scale=0.5]{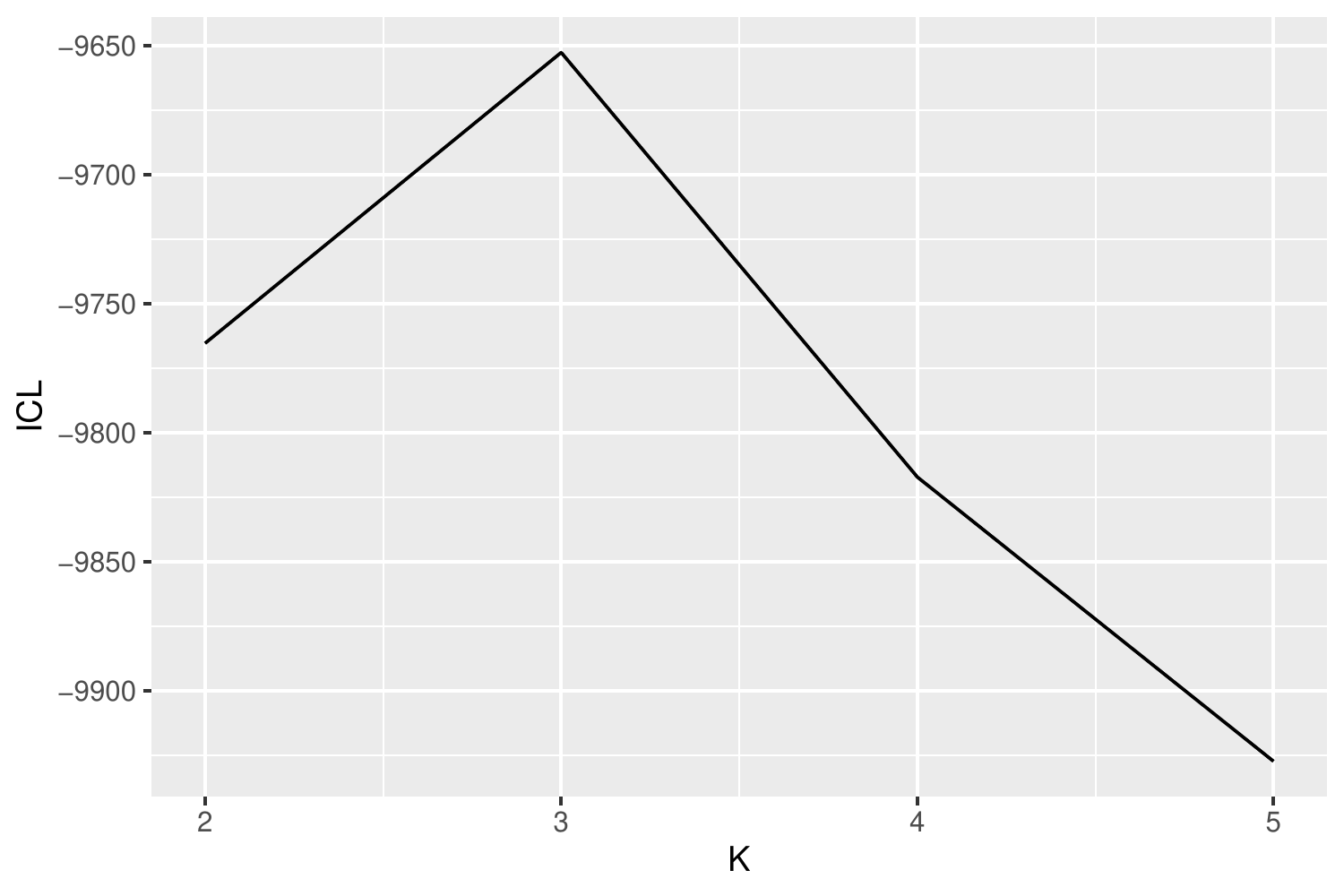}    
    \includegraphics[scale=0.5]{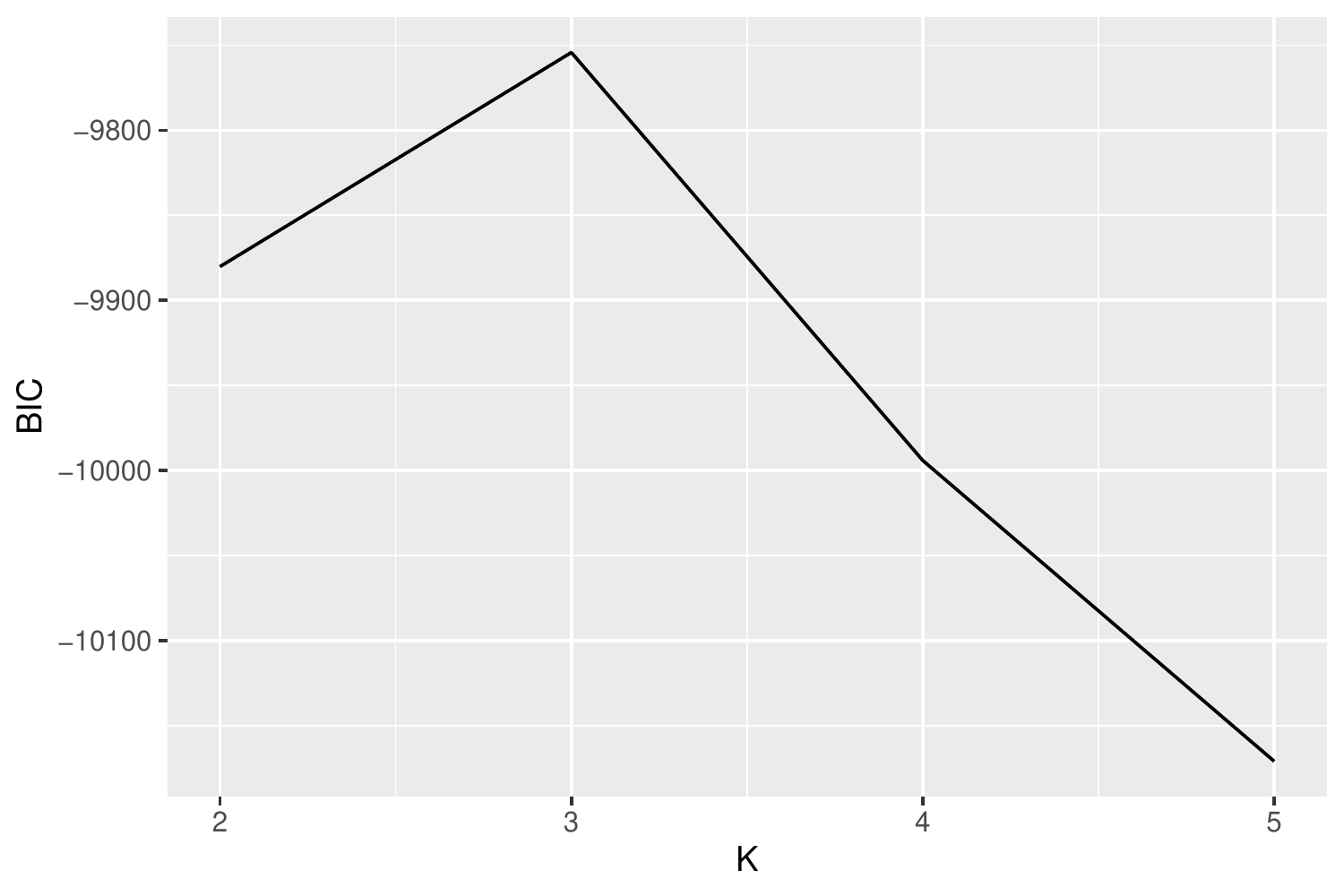}    
\end{figure}

Options \texttt{'ICL'} and \texttt{'BIC'} enables to visualize the evolution of the criterion with respect to the number of clusters $K$.

\begin{verbatim}
RMMplot(Result,graph=c('Profiles'))
\end{verbatim}

\begin{figure}[H]
\centering
    \includegraphics[scale=0.5]{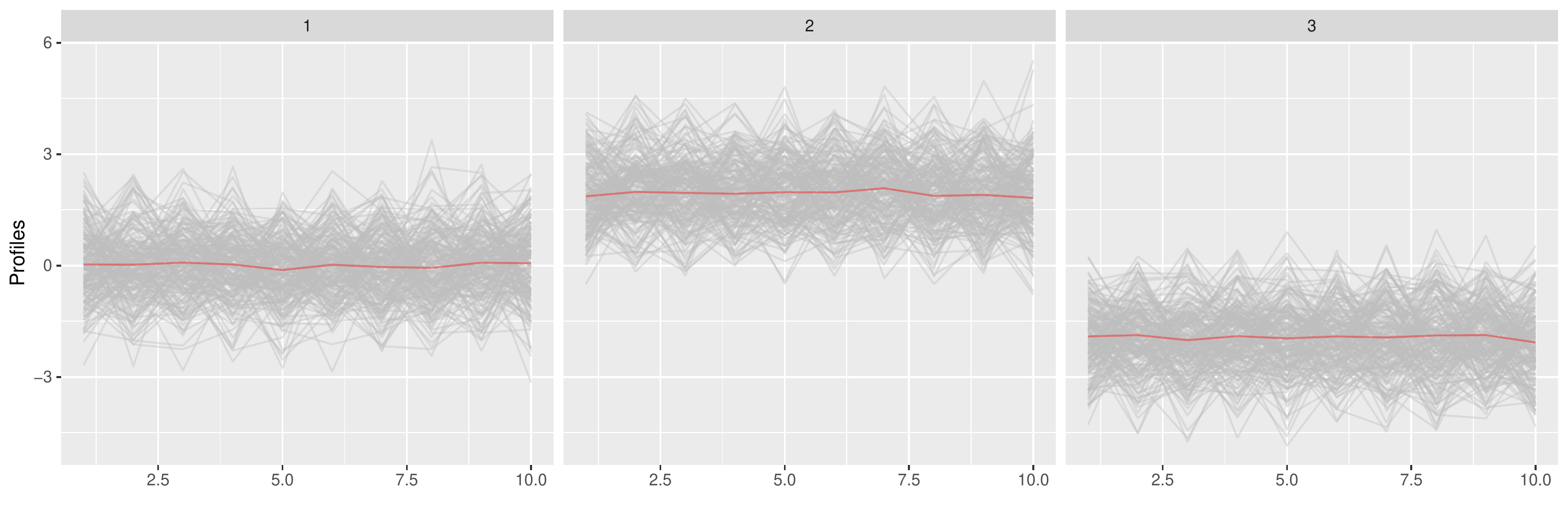}    
\end{figure}

Option \texttt{'Profiles'} allows to visualize data points in dimensions higher than $3$. More precisely, we represent data as curves that we call "profiles", gathered it by cluster, and represented the centers of the groups in red.

\begin{verbatim}
RMMplot(Result,graph=c('Uncertainty'))
\end{verbatim}

\begin{figure}[H]
\centering
    \includegraphics[scale=0.5]{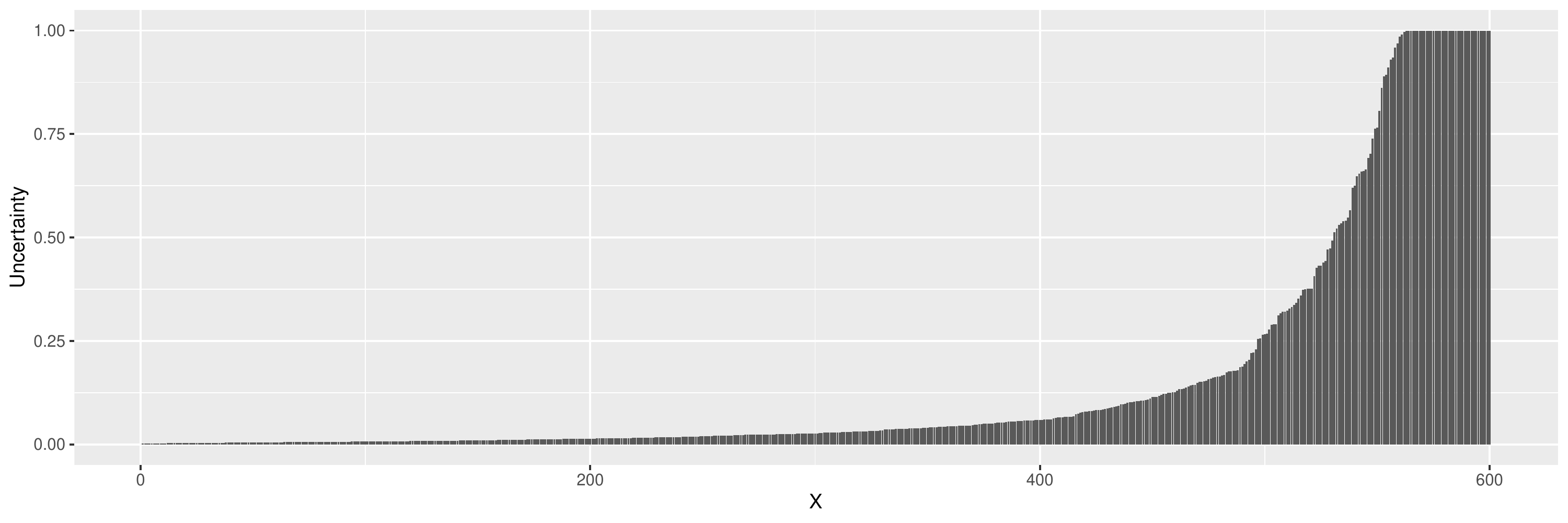}    
\end{figure}

Option \texttt{'Uncertainty'} enables to visualize the uncertainty of classification of the data.

\subsection{Additional simulation results for the estimation of the variance}
\label{app::variance}
In this section, we provide analogous tables as in Section \ref{sec:simVariance}, but with analogous calculus budget. More precisely, for gradient and fix point method, we consider a sample size of $N= 2000$ for the Monte-Carlo method, while for Robbins-Monro procedure, we consider a sample size of $N= 100000$. One can remark that Robbins-Monro  procedure provide better results for analogous computational budgets. Nevertheless, the difference is very slight because of the sample size for estimating the MCM is moderate, and the principal error seems to come from this "bad" estimation.

\begin{table}[H]
\centering
\begin{tabular}[b]{cc|rrrrrrr}
& \rotatebox[origin=r]{360}{$\delta$ ($\%$)}  & \rotatebox{270}{FixPoint (R)}    & \rotatebox{270}{FixPoint (W)}    & \rotatebox{270}{Gradient (R)}    & \rotatebox{270}{Gradient (W)}    & \rotatebox{270}{Robbins (R)}    & \rotatebox{270}{Robbins (W)}    & \rotatebox[origin=l]{270}{Variance}       \\  
   \hline
\multirow{8}{*}{\rotatebox{90}{ {($a$): $U$}}}& 0  & 0.35 & 0.29 & 0.36 & 0.31 & 0.19 & 0.18 & \textbf{0.11} \\ 
&   2 & 0.41 & 0.37 & 0.41 & 0.31 & 0.21 &\textbf{ 0.19} & 34 \\ 
&   3 & 0.43 & 0.36 & 0.37 & 0.44 & 0.26 & \textbf{0.24} & 75 \\ 
&   5 & 0.67 & 0.58 & 0.63 & 0.57 & 0.48 & \textbf{0.44} & 220 \\ 
&   9 & 1.27 & 1.25 & 1.33 & 1.18 & 1.13 & \textbf{1.08} & 690 \\ 
&   16 & 4.06 & 3.82 & 3.89 & \textbf{3.69} & 3.77 & 3.70 & $2.10^{3}$ \\ 
&  28 & 17.09 & 17.53 & 16.84 & \textbf{16.12} & 16.84 & 16.88 & $7.10^{3}$ \\ 
 & 50 & 152.32 & 158.95 & \textbf{128.75} & 136.53 & 147.27 & 153.80 & $2.10^{4}$ \\ 
\hline
\multirow{8}{*}{\rotatebox{90}{ {($b$): $T_1$}}} &   
0 & 0.27 & 0.24 & 0.31 & 0.27 & 0.15 & 0.13 & \textbf{0.09} \\ 
&   2 & 0.32 & 0.31 & 0.32 & 0.32 & 0.19 & \textbf{0.17} & $3.10^{8}$ \\ 
&   3 & 0.29 & 0.29 & 0.30 & 0.28 & 0.17 & \textbf{0.16} & $10^{13}$ \\ 
&   5 & 0.37 & 0.35 & 0.33 & 0.34 & 0.19 & \textbf{0.18} & $10^{12}$ \\ 
&   9 & 0.46 & 0.43 & 0.44 & 0.43 & \textbf{0.29} & \textbf{0.29} & $2.10^{12}$ \\ 
&   16 & 0.71 & 0.75 & 0.81 & 0.78 & 0.62 & \textbf{0.61} & $3.10^{11}$ \\ 
&   28 & 1.85 & 1.85 & 1.67 & 1.82 & 1.62 & \textbf{1.62} & $10^{14}$ \\ 
&   50 & 5.39 & 5.42 & 5.48 & 5.48 & 5.23 & \textbf{5.19} & $10^{11}$ \\ 
   \hline
\multirow{8}{*}{\rotatebox{90}{ {($e$): $T_2$}}} &    0 & 0.37 & 0.33 & 0.26 & 0.24 & 0.15 & 0.14 & \textbf{0.09} \\ 
&   2 & 0.38 & 0.35 & 0.33 & 0.33 & 0.19 & 0.18 & \textbf{0.10} \\ 
&   3 & 0.33 & 0.29 & 0.33 & 0.31 & 0.20 & \textbf{0.19} & 3.3 \\ 
&   5 & 0.47 & 0.54 & 0.49 & 0.52 & \textbf{0.34} & \textbf{0.34} & 45 \\ 
&   9 & 1.22 & 0.90 & 1.03 & 1.01 & \textbf{0.87} & 0.89 & 110 \\ 
&   16 & 2.04 & 2.11 & 2.15 & 2.14 & \textbf{1.98} & 2.06 & 57 \\ 
&   28 & \textbf{5.37} & 5.67 & 5.49 & 5.72 & 5.46 & 5.66 & 810 \\ 
&   50 & 14.97 & 15.66 & \textbf{14.90} & 15.38 & 15.11 & 15.31 & 940 \\ 
   \hline
\end{tabular}
\caption{\label{gaussiancaseapp}
 {Multivariate Gaussian case: Mean quadratic error} of the estimates of the variance for the different methods and for different contamination scenarios and fractions $\delta$.}
\end{table}

\begin{table}[H]
\centering
\begin{tabular}{cc|rrrrrrr}
& \rotatebox[origin=r]{360}{$\delta$ ($\%$)}  & \rotatebox{270}{FixPoint (R)}    & \rotatebox{270}{FixPoint (W)}    & \rotatebox{270}{Gradient (R)}    & \rotatebox{270}{Gradient (W)}    & \rotatebox{270}{Robbins (R)}    & \rotatebox{270}{Robbins (W)}    & \rotatebox[origin=l]{270}{Variance}       \\  
   \hline
\multirow{8}{*}{\rotatebox{90}{{($a$): $U$}}} &   0 & 0.25 & 0.23 & 0.16 & \textbf{0.15} & 0.17 & \textbf{0.15} & 45.82 \\ 
&   2 & 0.42 & 0.37 & 0.30 & \textbf{0.24} & 0.32 & \textbf{0.24} & 35.23 \\ 
&   3 & 0.49 & 0.39 & 0.34 & \textbf{0.27} & 0.36 & 0.29 & 285.08 \\ 
&   5 & 0.75 & 0.70 & 0.68 & \textbf{0.59} & 0.71 & 0.62 & 211.14 \\ 
&   9 & 2.09 & 1.88 & 1.82 & \textbf{1.61} & 1.96 & 1.72 & 679.89 \\ 
&   16 & 7.23 & 6.20 & \textbf{6.08} & 5.64 & 6.59 & \textbf{6.08}  & $2.10^{ 3}$ \\ 
&   28 & 9.36 & 27.09 & 26.49 & \textbf{24.74} & 29.40 & 27.69 & $6.10^{3}$  \\ 
&   50 & 374.82 & 363.92 & 261.54 & \textbf{256.83} & 371.40 & 361.54 & $2.10^{4}$ \\ 
\hline
\multirow{8}{*}{\rotatebox{90}{{($b$): $T_1$}}} &   0 & 0.26 & 0.24 & 0.19 & \textbf{0.17} & 0.19 & \textbf{0.17} & 4.29 \\ 
&   2 & 0.37 & 0.27 & 0.21 & \textbf{0.18} & 0.21 & 0.19 & $7.10^{11}$ \\ 
&   3 & 0.36 & 0.36 & 0.24 & \textbf{0.21} & 0.26 & 0.22 & $1.10^{8}$ \\ 
&   5 & 0.48 & 0.38 & 0.30 & \textbf{0.26} & 0.32 & 0.27 & $6.10^{10}$ \\ 
&   9 & 0.77 & 0.75 & 0.67 & \textbf{0.57} & 0.70 & 0.59 & $9.10^{10}$ \\ 
&   16 & 2.17 & 1.80 & 1.89 & \textbf{1.69} & 1.96 & 1.72 & $2.10^{13}$ \\ 
  & 28 & 6.44 & 6.04 & 5.99 & \textbf{5.54} & 6.22 & 5.82 & $7.10^{10}$ \\ 
 &  50 & 30.74 & 29.61 & 28.60 & \textbf{27.70} & 30.85 & 29.56 & $5.10^{14}$\\
  \hline
\multirow{8}{*}{\rotatebox{90}{{($e$): $T_2$}}} &   0 &0.25 & 0.27 & 0.16 & \textbf{0.15} & 0.17 & \textbf{0.15} & 8.57  \\ 
&   2 & 0.28 & 0.27 & 0.20 & \textbf{0.18} & 0.20 &\textbf{ 0.18} & $4.10^{3}$  \\ 
&   3 & 0.29 & 0.25 & 0.20 & \textbf{0.18} & 0.20 & \textbf{0.18} & 62.81 \\ 
 &  5 &  0.33 & 0.33 & \textbf{0.21} & \textbf{0.21} & 0.22 & \textbf{0.21} & 5.69  \\ 
 &  9 & 0.43 & 0.48 & 0.36 & 0.34 & 0.35 & \textbf{0.33} & 387.11 \\ 
 &  16 &0.86 & 0.77 & 0.68 & 0.68 & 0.69 & \textbf{0.65} & 267.72\\ 
 &  28 & 1.96 & 1.92 & 1.91 & 1.85 & 1.90 & \textbf{1.82} & 164.25 \\ 
 &  50 & 6.66 & 6.41 & 6.59 & 6.48 & 6.59 & \textbf{6.41} & 279.53 \\ 
   \hline
\end{tabular}
\caption{ {Multivariate Student case: Mean quadratic error} of the estimates of the variance for the different methods and for different contamination scenarios and fractions $\delta$.}
\end{table}

\subsection{Additional simulation results for Mixture models}
\label{app:simResults}

Figures \ref{fig:simGMM1500} and \ref{fig:simTMM1500} provide the simulation results with smaller sample size, namely $n_k = 100$ observations per clusters, that is $n = 300$ observations in total. Figure \ref{fig:simGMM300} is the counterpart of Figure \ref{fig:simGMM1500}, and Figure \ref{fig:simTMM300} is this of Figure \ref{fig:simTMM1500}.

\begin{figure}[ht]
  \centering
    \begin{tabular}{c|m{.2\textwidth}m{.2\textwidth}|m{.2\textwidth}m{.2\textwidth}}
      & 
      \multicolumn{2}{c|}{Classification} & 
      \multicolumn{2}{c}{Parameter estimation} \\ 
      &
      \multicolumn{1}{c}{ARI} & \multicolumn{1}{c|}{$\widehat{K}$} & 
      \multicolumn{1}{c}{$MSE(\mu)$} & \multicolumn{1}{c}{$MSE(\Sigma)$} \\ 
      \hline
      ($a$) &
      \includegraphics[width=.2\textwidth, trim=10 10 10 50, clip=]{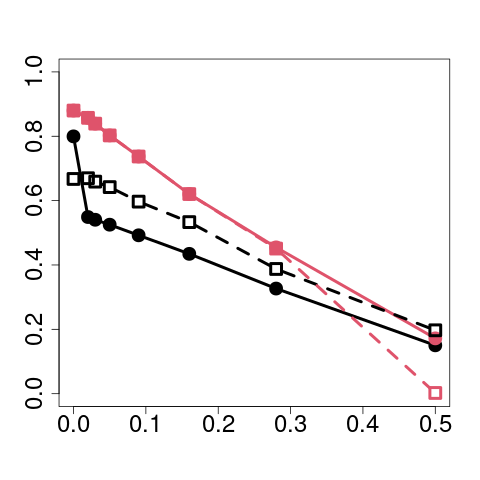} &
      \includegraphics[width=.2\textwidth, trim=10 10 10 50, clip=]{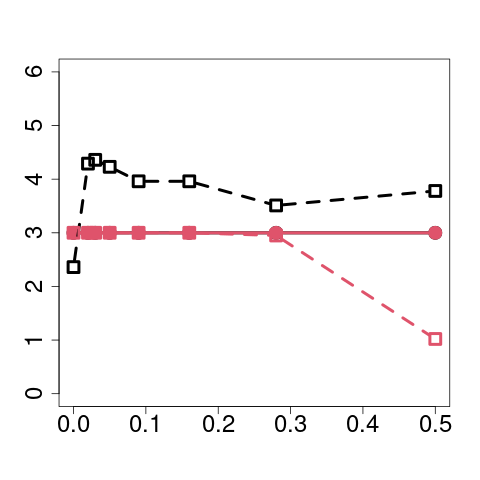} &
      \includegraphics[width=.2\textwidth, trim=10 10 10 50, clip=]{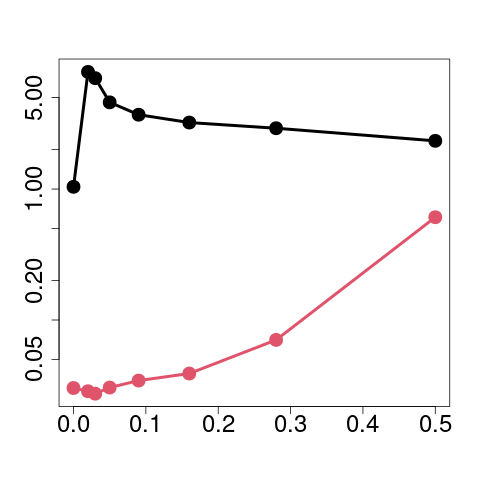} &
      \includegraphics[width=.2\textwidth, trim=10 10 10 50, clip=]{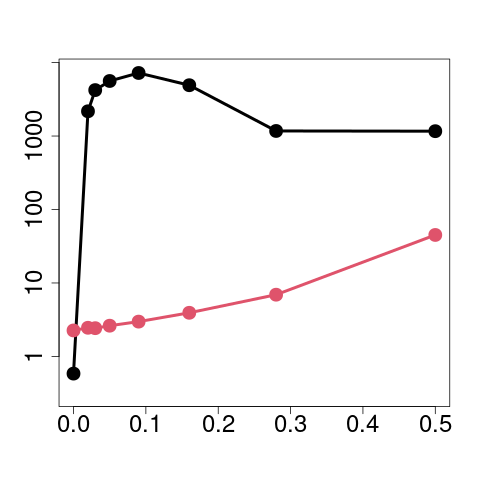} \\
      ($b$) &
      \includegraphics[width=.2\textwidth, trim=10 10 10 50, clip=]{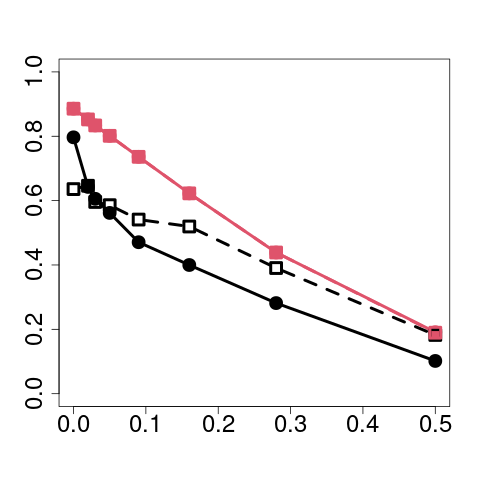} &
      \includegraphics[width=.2\textwidth, trim=10 10 10 50, clip=]{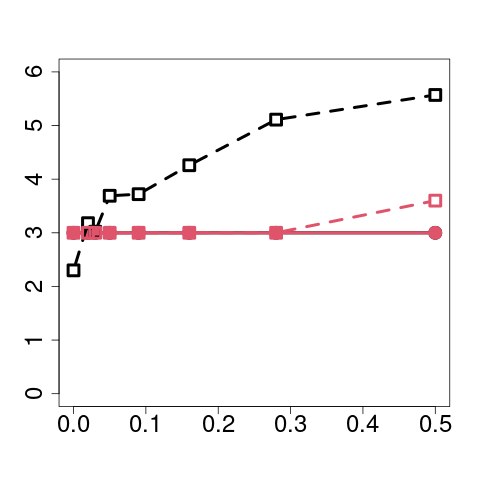} &
      \includegraphics[width=.2\textwidth, trim=10 10 10 50, clip=]{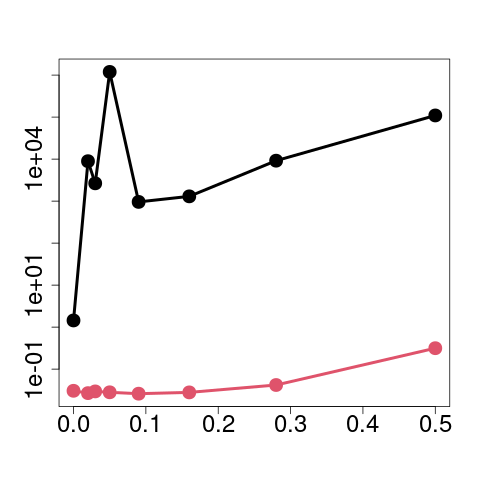} &
      \includegraphics[width=.2\textwidth, trim=10 10 10 50, clip=]{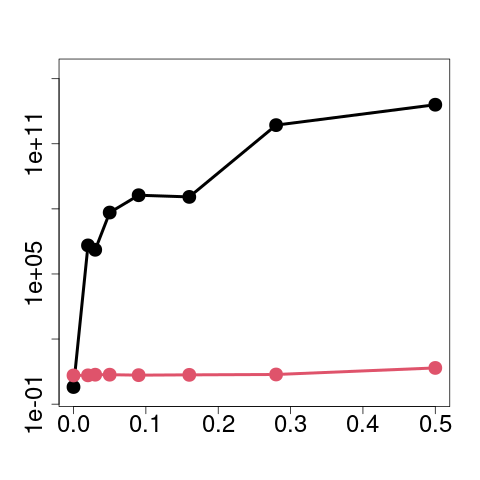} \\
      ($c$) &
      \includegraphics[width=.2\textwidth, trim=10 10 10 50, clip=]{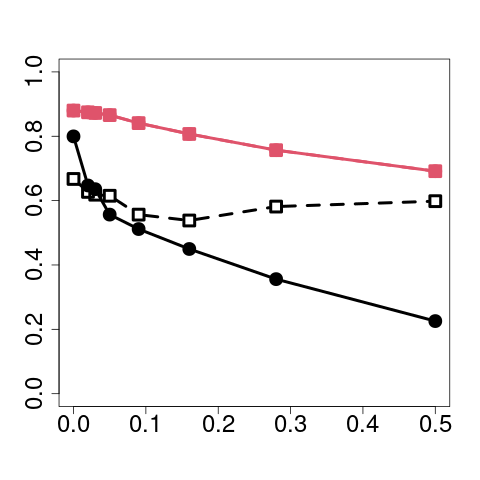} &
      \includegraphics[width=.2\textwidth, trim=10 10 10 50, clip=]{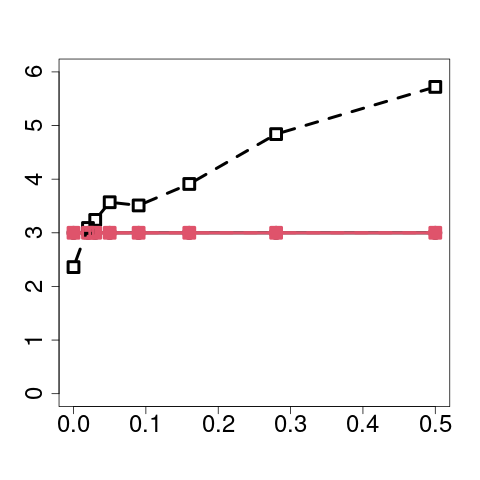} &
      \includegraphics[width=.2\textwidth, trim=10 10 10 50, clip=]{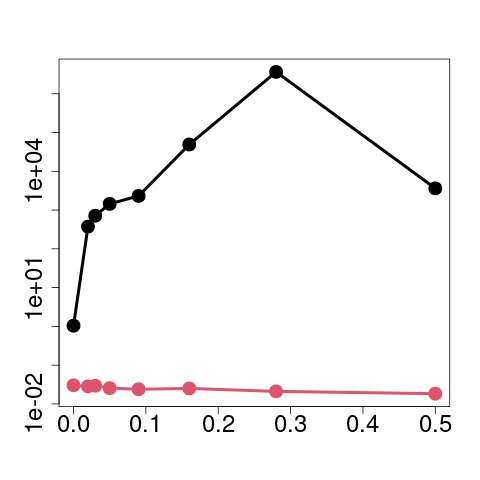} &
      \includegraphics[width=.2\textwidth, trim=10 10 10 50, clip=]{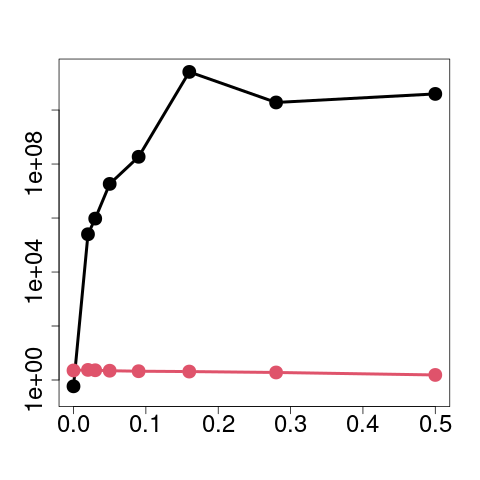} \\
      ($d$) &
      \includegraphics[width=.2\textwidth, trim=10 10 10 50, clip=]{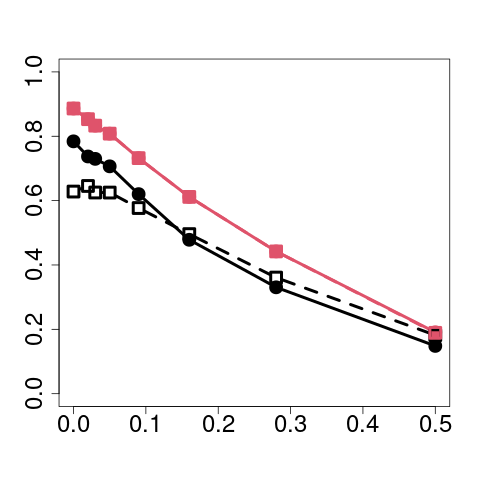} &
      \includegraphics[width=.2\textwidth, trim=10 10 10 50, clip=]{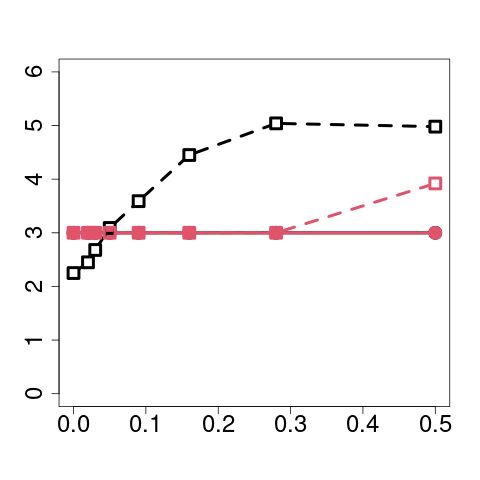} &
      \includegraphics[width=.2\textwidth, trim=10 10 10 50, clip=]{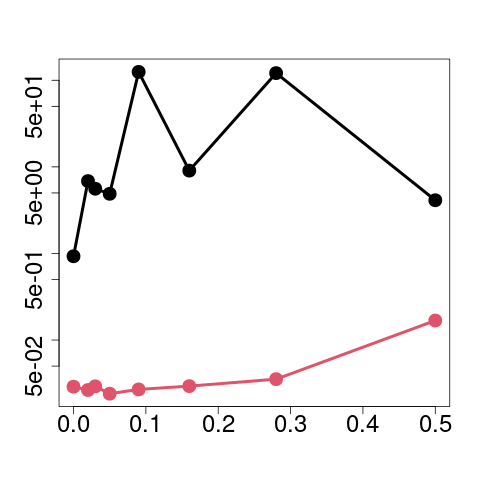} &
      \includegraphics[width=.2\textwidth, trim=10 10 10 50, clip=]{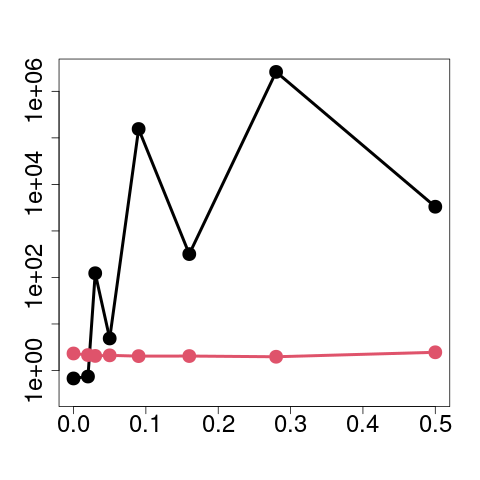} \\
      ($e$) &
      \includegraphics[width=.2\textwidth, trim=10 10 10 50, clip=]{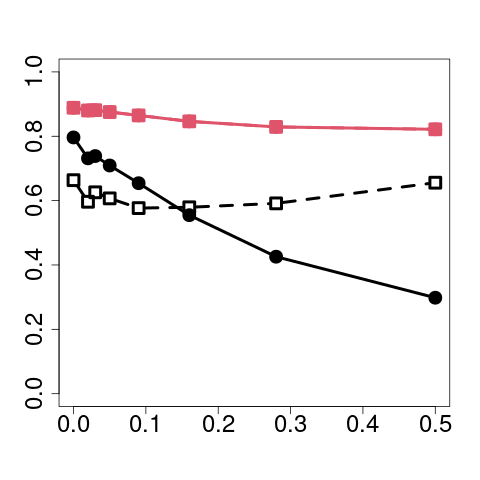} &
      \includegraphics[width=.2\textwidth, trim=10 10 10 50, clip=]{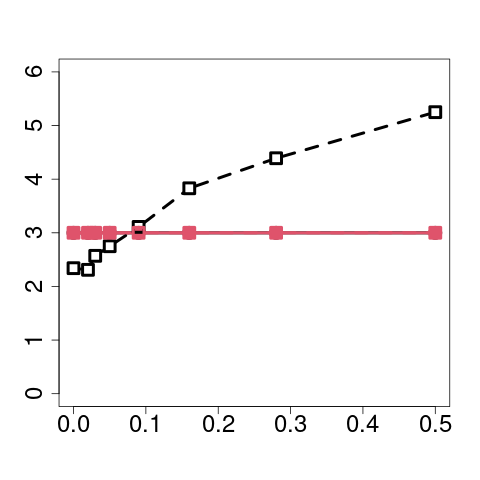} &
      \includegraphics[width=.2\textwidth, trim=10 10 10 50, clip=]{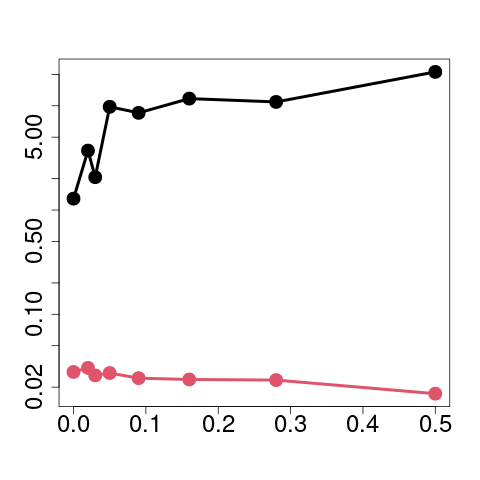} &
      \includegraphics[width=.2\textwidth, trim=10 10 10 50, clip=]{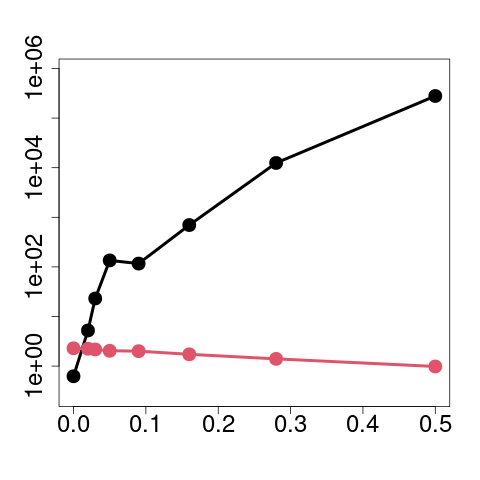} 
    \end{tabular}
    \caption{Gaussian mixture model: classification accuracy ($ARI$), estimated number of clusters $\widehat{K}$, estimation error fu the mean ($MSE(\mu)$) and for the variance ($MSE(\Sigma)$) for scenarios ($a$) to ($e$), with $n_k = 100$ observation in each of the $K^*$ clusters ($n = 300$). Same legend as Figure \ref{fig:simGMM1500}. \label{fig:simGMM300}}
\end{figure}

\begin{figure}[ht]
  \centering
    \begin{tabular}{c|m{.2\textwidth}m{.2\textwidth}|m{.2\textwidth}m{.2\textwidth}}
      & 
      \multicolumn{2}{c|}{Classification} & 
      \multicolumn{2}{c}{Parameter estimation} \\ 
      &
      \multicolumn{1}{c}{ARI} & \multicolumn{1}{c|}{$\widehat{K}$} & 
      \multicolumn{1}{c}{$MSE(\mu)$} & \multicolumn{1}{c}{$MSE(\Sigma)$} \\ 
      \hline
      ($a$) &
      \includegraphics[width=.2\textwidth, trim=10 10 10 50, clip=]{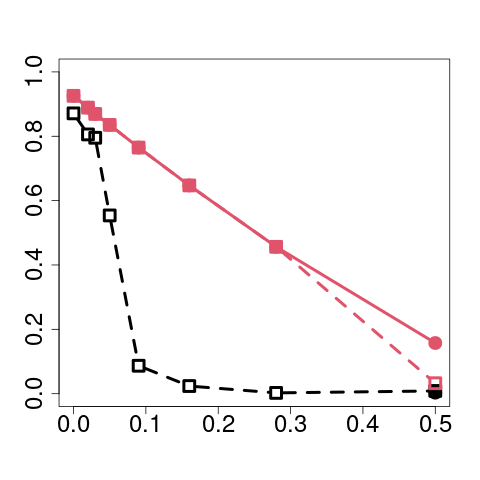} &
      \includegraphics[width=.2\textwidth, trim=10 10 10 50, clip=]{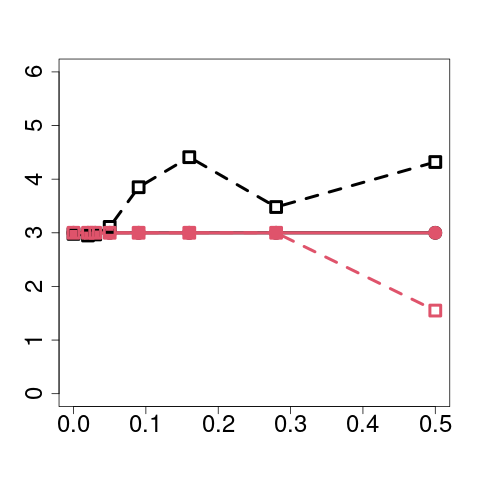} &
      \includegraphics[width=.2\textwidth, trim=10 10 10 50, clip=]{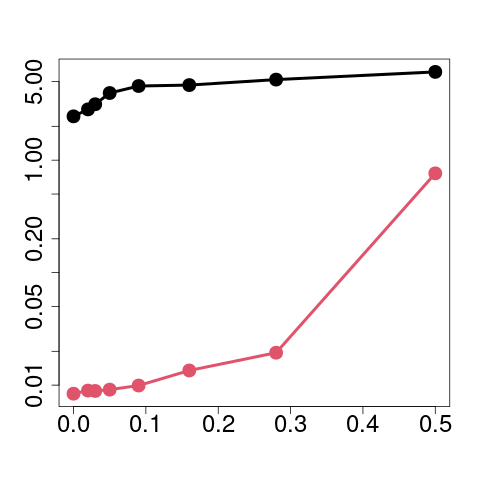} &
      \includegraphics[width=.2\textwidth, trim=10 10 10 50, clip=]{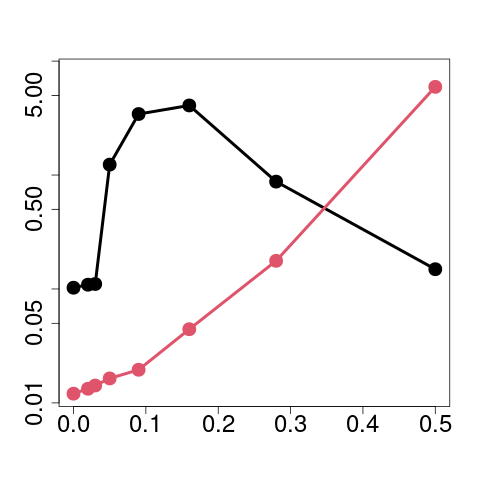} \\
      ($b$) &
      \includegraphics[width=.2\textwidth, trim=10 10 10 50, clip=]{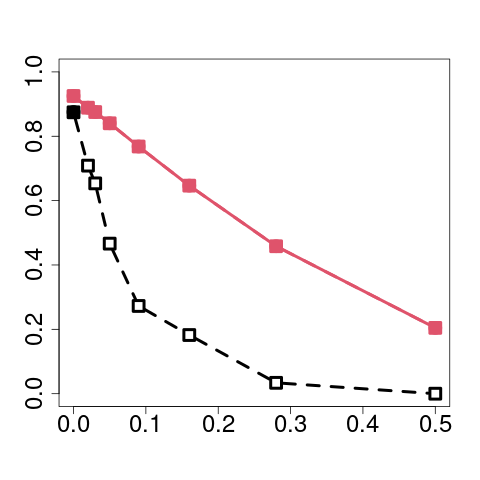} &
      \includegraphics[width=.2\textwidth, trim=10 10 10 50, clip=]{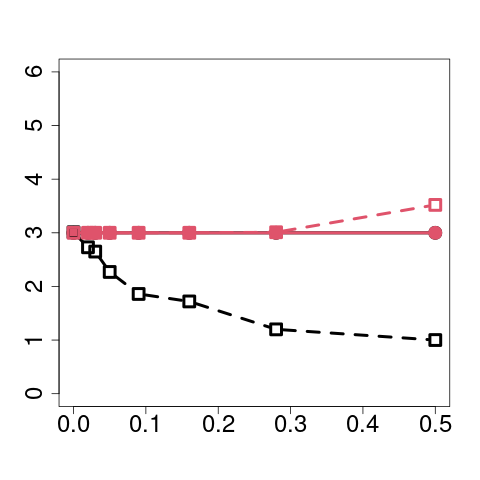} &
      \includegraphics[width=.2\textwidth, trim=10 10 10 50, clip=]{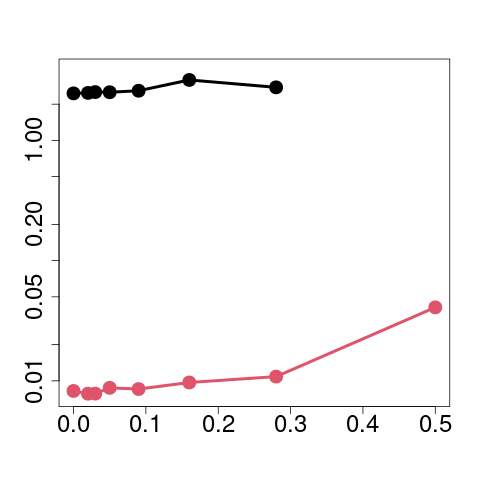} &
      \includegraphics[width=.2\textwidth, trim=10 10 10 50, clip=]{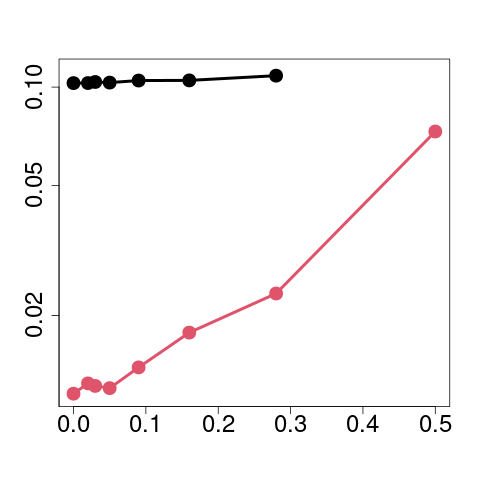} \\
      ($c$) &
      \includegraphics[width=.2\textwidth, trim=10 10 10 50, clip=]{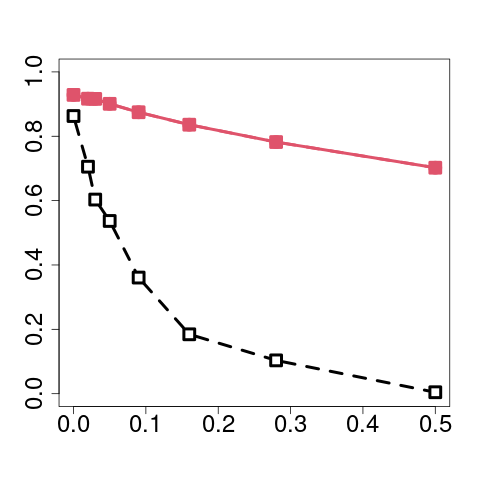} &
      \includegraphics[width=.2\textwidth, trim=10 10 10 50, clip=]{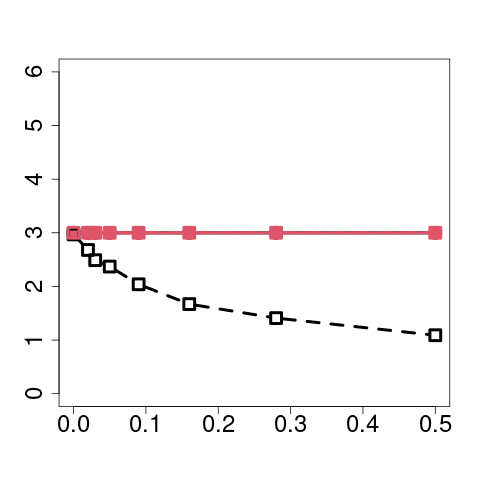} &
      \includegraphics[width=.2\textwidth, trim=10 10 10 50, clip=]{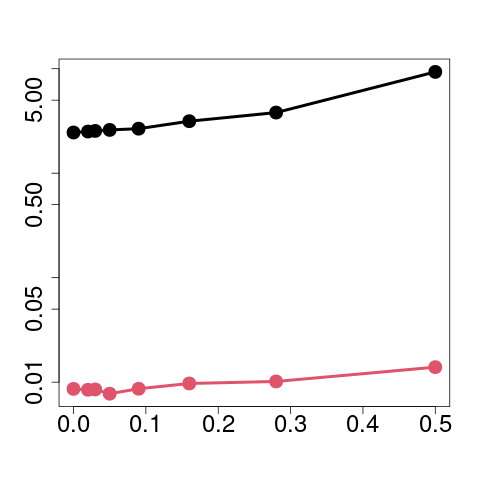} &
      \includegraphics[width=.2\textwidth, trim=10 10 10 50, clip=]{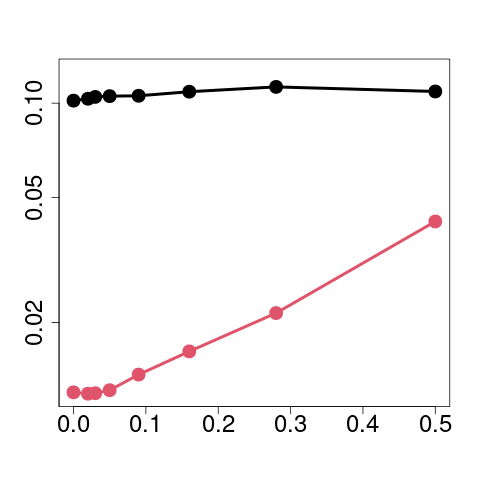} \\
      ($d$) &
      \includegraphics[width=.2\textwidth, trim=10 10 10 50, clip=]{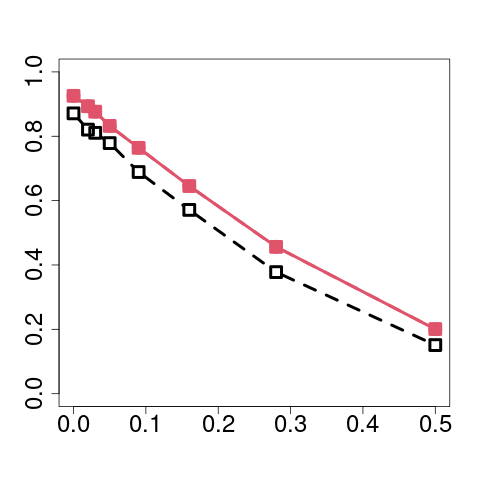} &
      \includegraphics[width=.2\textwidth, trim=10 10 10 50, clip=]{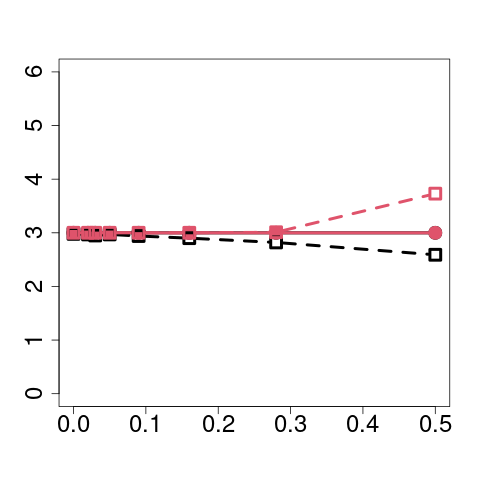} &
      \includegraphics[width=.2\textwidth, trim=10 10 10 50, clip=]{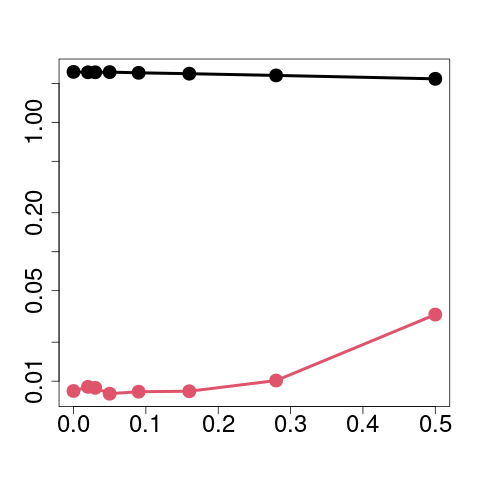} &
      \includegraphics[width=.2\textwidth, trim=10 10 10 50, clip=]{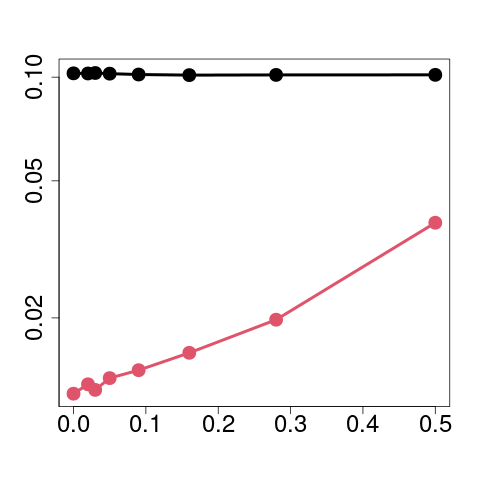} \\
      ($e$) &
      \includegraphics[width=.2\textwidth, trim=10 10 10 50, clip=]{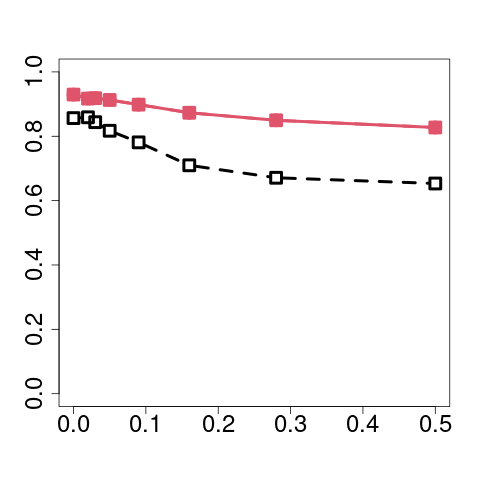} &
      \includegraphics[width=.2\textwidth, trim=10 10 10 50, clip=]{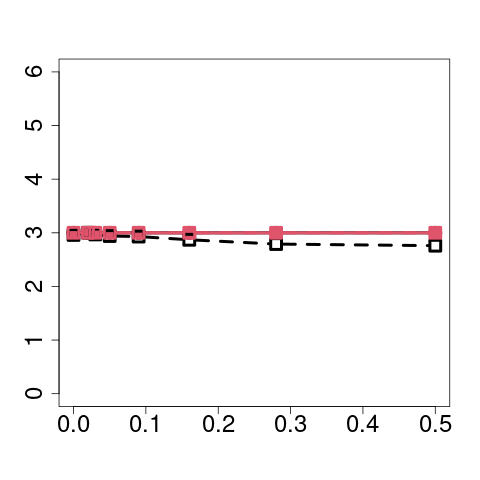} &
      \includegraphics[width=.2\textwidth, trim=10 10 10 50, clip=]{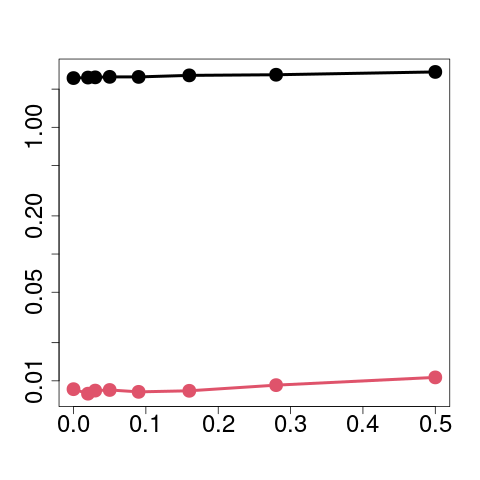} &
      \includegraphics[width=.2\textwidth, trim=10 10 10 50, clip=]{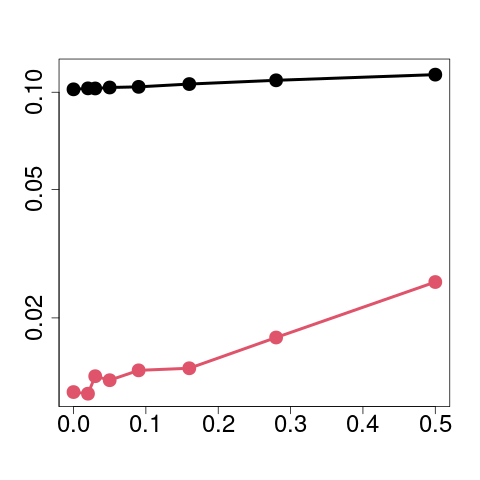} 
    \end{tabular}
    \caption{Student mixture model: classification accuracy ($ARI$), estimated number of clusters $\widehat{K}$, estimation error fu the mean ($MSE(\mu)$) and for the variance ($MSE(\Sigma)$) for scenarios ($a$) to ($e$), with $n_k = 100$ observation in each of the $K^*$ clusters ($n = 300$). Same legend as Figure \ref{fig:simGMM1500}. \label{fig:simTMM300}}
\end{figure}









\FloatBarrier

Figure {\ref{fig:simTMM1500fail} displays the proportion of simulations for which the inference algorithm for the inference of Student mixtures (teigen or RTMM) failed to converge.

\begin{figure}[ht]
  \centering
    \begin{tabular}{ccc}
      ($a$) & ($b$) & ($c$) \\
      \includegraphics[width=.2\textwidth, trim=10 10 10 50, clip=]{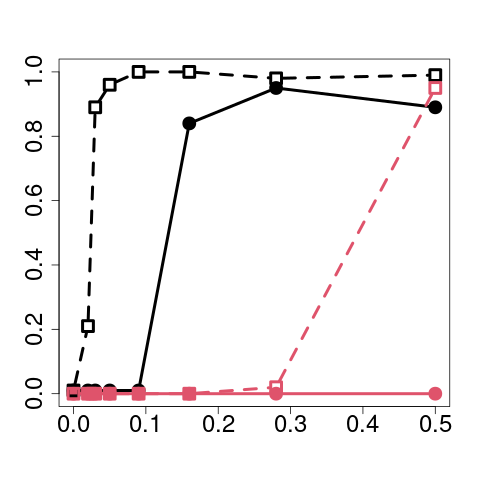} &
      \includegraphics[width=.2\textwidth, trim=10 10 10 50, clip=]{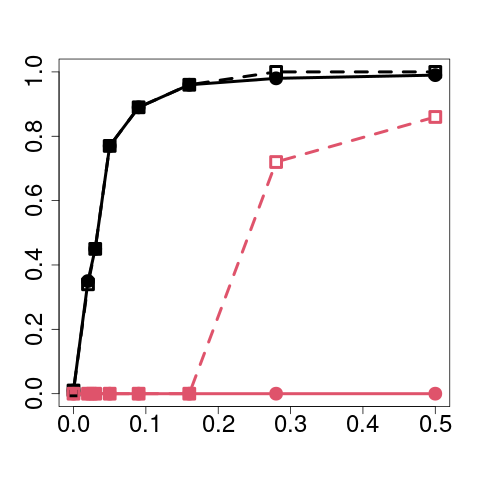} &
      \includegraphics[width=.2\textwidth, trim=10 10 10 50, clip=]{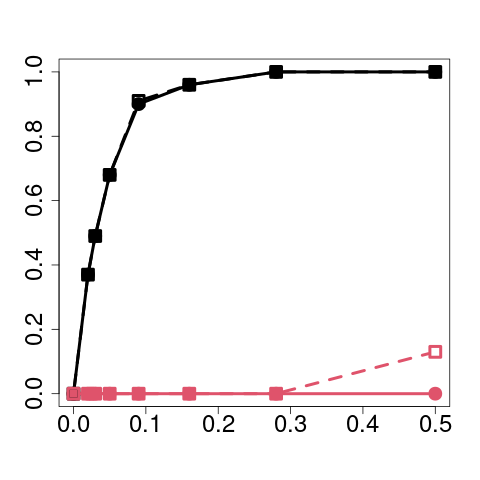} \\
      \hline
      & ($d$) & ($e$) \\
      & 
      \includegraphics[width=.2\textwidth, trim=10 10 10 50, clip=]{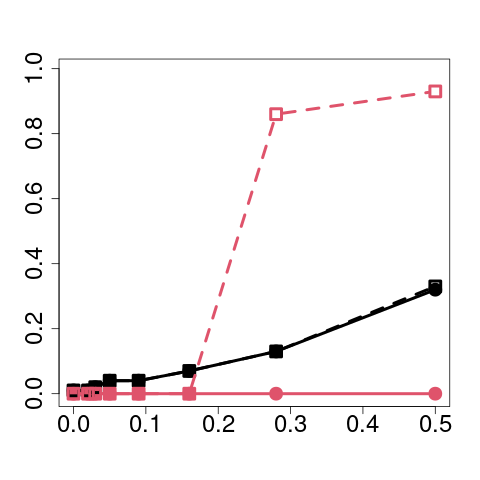} &
      \includegraphics[width=.2\textwidth, trim=10 10 10 50, clip=]{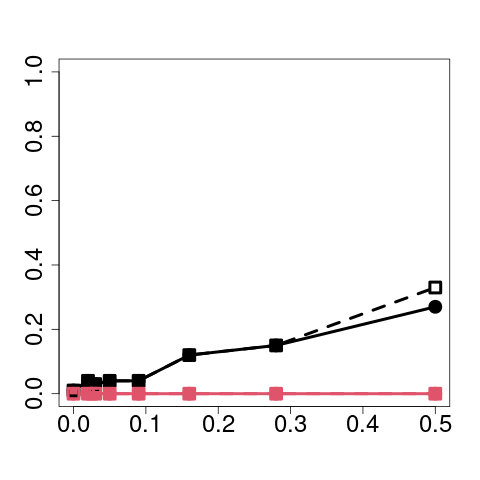}
    \end{tabular}
    \caption{Student mixture model: proportion of non-convergence of the inference algorithms for scenarios ($a$) to ($e$), with $n_k = 500$ observation in each of the $K^*$ clusters ($n = 1500$). Same legend as Figure \ref{fig:simGMM1500}. When the number $K$ is estimated (dotted lines), non-convergence that the algorithm did not converge for at least on value of $K$. \label{fig:simTMM1500fail}}
\end{figure}

\end{document}